\documentclass[journal]{IEEEtran}
\usepackage{amssymb}
\usepackage{amsmath}
\usepackage{graphicx}
\usepackage{multirow}
\usepackage{enumerate}
\usepackage{subfig}
\usepackage{comment}

\newcommand{\E}[1]{\textrm{E}\left[#1\right]}
\newcommand{\var}[1]{\textrm{var}\left\{#1\right\}}
\newcommand{\Ebr}[1]{\textrm{E}\left\{#1\right\}}

\newcommand{\vct}[1]{\boldsymbol{\mathbf{#1}}}
\newcommand{\mat}[1]{\boldsymbol{\mathbf{#1}}}

\newcommand{\tr}{\textrm{T}}
\newcommand{\tra}{\mathrm{T}}

\newcommand{\minimaxmse}{\rho}

\newcommand{\tm}[1]{\textrm{#1}}
\newcommand{\ti}[1]{\textit{#1}}
\newcommand{\nn}{\nonumber}
\newcommand{\convest}{\xi}

\newcommand{\dms}{ d_{1} }
\newcommand{\dsm}{ d_{2} }
\newcommand{\wms}{ w_{1} }
\newcommand{\wsm}{ w_{2} }
\newcommand{\dmin}{d}
\newcommand{\fms}[1]{ f_{1}\left( #1 \right) }
\newcommand{\fsm}[1]{ f_{2}\left( #1 \right) }

\newcommand{\legendwidth}{140px}
\newcommand{\perfcurveswidth}{0.43\textwidth}

\newtheorem{theorem}{Theorem}

\newtheorem{lemma}{Lemma}

\newtheorem{corollary}{Corollary}
\newtheorem{definition}{Definition}

\begin{document}

\title{Minimax Optimum Estimators for Phase Synchronization in IEEE 1588}
\author{Anand~Guruswamy,~\IEEEmembership{Student Member,~IEEE,}
        Rick~S.~Blum,~\IEEEmembership{Fellow,~IEEE,}
        Shalinee~Kishore~\IEEEmembership{Member,~IEEE}
		and~Mark~Bordogna
\thanks{A. Guruswamy, R.S. Blum and S. Kishore are with the Department of Electrical and Computer Engineering, Lehigh University, Bethlehem, PA, USA (e-mail: \{asg210, rblum, skishore\}@lehigh.edu).}
\thanks{M. Bordogna is with Intel Corporation at Allentown, PA, USA (e-mail: mark.bordogna@intel.com).}
}
\maketitle

\begin{abstract}
The IEEE 1588 protocol has received recent interest as a means of delivering sub-microsecond level clock phase synchronization over packet-switched mobile backhaul networks.
Due to the randomness of the end-to-end delays in packet networks, the recovery of clock phase from packet timestamps in IEEE 1588 must be treated as a statistical estimation problem.
A number of estimators for this problem have been suggested in the literature, but little is known about the best achievable performance. 
In this paper, we describe new minimax estimators for this problem, that are optimum in terms of minimizing the maximum mean squared error over all possible values of the unknown parameters.
Minimax estimators that utilize information from past timestamps to improve accuracy are also introduced.
Simulation results indicate that significant performance gains over conventional estimators can be obtained via such optimum processing techniques.
These minimax estimators also provide fundamental limits on the performance of phase offset estimation schemes.
\end{abstract}
\section{Introduction}

In modern 4G cellular networks, precise synchronization between base stations is critical for ensuring seamless handovers, reducing interference and improving capacity.
Given the high cost and effort associated with global positioning system (GPS) based synchronization, carriers often find it preferable to deliver timing via the mobile backhaul network.
Since these backhaul networks are typically packet switched in nature, a popular timing approach \cite{ouellette2011using} is to use SyncE \cite{g8261}\cite{g8262} for frequency synchronization and IEEE 1588 \cite{ieee1588} for phase synchronization.
The topic of phase synchronization is the focus of this paper.
A related requirement arising from 4G LTE (Long Term Evolution) networks is that neighboring base stations must be synchronized to within $1.25$ $\mu s$ of absolute phase error, to ensure efficient operation in the time division duplexing mode.

In the IEEE 1588 precision time protocol (PTP), a master and a slave node exchange a series of packets to achieve phase synchronization. 
Packets traveling between the master and the slave encounter several intermediate network nodes such as switches or routers, accumulating random queuing delays at each node. The problem of finding the slave's phase offset from the timestamps of the exchanged packets, while combating the random queuing delays, is referred to as \textit{phase offset estimation} (POE).
The PTP standard and related literature prescribe the use of simple estimators such as the sample mean, minimum and maximum filters for POE. 
Several recent papers \cite{hadzic2010adaptive}\nocite{anyaegbu2012sample,murakami2011packet, peng2013clock, giorgi2011performance,bletsas2005evaluation}--\cite{iantosca1588synchronizing} have studied methods to improve the performance of these filters, especially in the presence of large queuing delays due to high network loads.
However, it is not well understood as to how close these POE schemes come to achieving the best possible  performance.

In this paper, we derive optimum estimators for the problem of POE, which, to our knowledge, have not been described previously in literature.
To this end, in Section  \ref{sec:PDV_SysModel} we begin by modeling POE as a non-Bayesian estimation problem.
Specifically, we treat the phase offset as an unknown deterministic parameter to be estimated from timestamps that are also affected by the fixed delays along the forward and reverse network paths. 
We then consider three observation models, with the degree of information available about the fixed delays varying between these models.
The first model assumes complete knowledge of both the fixed delays, while the second model assumes only that the difference between the fixed delays, i.e. the \textit{delay asymmetry}, is known. The third model assumes known delay asymmetry, as well as the availability of additional past observations which contain the same fixed delays but different phase offsets.
We show that POE under all three models falls under a general class of estimation problems known as \textit{vector location parameter problems}.
In Section \ref{sec:minimaxest}, for this general problem class, we derive the optimum estimator that minimizes the maximum mean squared error (\textit{maximum MSE}) over all values of the unkown parameters, and is hence \textit{minimax} optimum. 
This minimax estimator is an extension of the well-studied \textit{Pitman estimator} \cite{pitman1939estimation}, which is known to be minimax optimum for scalar location parameter problems.
Other properties of the minimax estimator, related to the estimation of linear combinations of parameters, are also derived.

In Section \ref{sec:poe_simplifications}, we simplify the general minimax estimator for the problem of POE under each observation model.
In Section \ref{sec:minimax_cross_traf}, using the properties of the minimax estimator derived in Section \ref{sec:minimaxest}, we show that under typical network assumptions, the MSE of the minimax estimator grows at least linearly with the number of intermediate nodes between the master and the slave.
Our simulations in section \ref{sec:simresults} compare the performance of the new minimax estimates against conventional estimators under several network conditions. 
Results indicate that there are several network scenarios where conventional estimation schemes fall significantly short of achieving  the maximum possible synchronization accuracy. 
Further, in asymmetric network traffic scenarios, we show that significant performance gains become available if we exploit information about fixed delays from past observations.

The results in this paper extend our previous works \cite{guru2014perf}\cite{guru2014perfconf}, where lower bounds on the maximum MSE of POE schemes under the second observation model were derived. In this paper, we address more observational models, provide the tightest lower bounds on the maximum MSE of POE schemes under each model, and also specify the estimators that achieve these lower bounds.

\section{System Model} \label{sec:PDV_SysModel}
Consider a scenario where the slave clock has a phase offset $\delta$ and zero frequency offset with respect to its master. 
To help the slave determine $\delta$, the IEEE 1588 PTP protocol allows a \textit{two-way message exchange} between the master and slave. 
Four timestamps are available to the slave after each two-way message exchange (for more details, see \cite{guru2014perf}):
\begin{enumerate}[(a)]
\item[$t_1$]: Time of transmission of \ti{SYNC} packet by master.
\item[$t_2$]: Time of reception of \ti{SYNC} packet at slave.
\item[$t_3$]: Time of transmission of \ti{DELAY\_REQ} packet by slave.
\item[$t_4$]: Time of reception of \ti{DELAY\_REQ}  packet by master.
\end{enumerate}

In order to estimate $\delta$, it is clearly sufficient for the slave to only retain the pair of timestamp differences 
\begin{align}
y_1 &= t_2-t_1 = \dms + \delta  \\
y_2 &= t_4-t_3 = \dsm - \delta 
\end{align}
Here $d_1$ and $d_2$ denote the end-to-end (ETE) network delays in the master-slave and slave-master directions, respectively. 
Assume for simplicity that a common network path is taken by all packets traveling between the master and the slave and vice-versa. 
Then each ETE delay receives contributions from three factors: 
\begin{enumerate}[(a)]
\item Constant propagation delays along network links between the master and the slave (or vice-versa).
\item Constant processing delays at intermediate nodes (such as switches or routers) along each network path.
\item Random queuing delays at intermediate nodes along each network path.
\end{enumerate}
Hence each ETE delay can be modeled as 
\begin{align}
\dms= \dms^{\min} + \wms, \qquad 
\dsm= \dsm^{\min} + \wsm
\end{align}
Here $\dms^{\min}$ and $\dsm^{\min}$ denote fixed delays corresponding to the sum of the constant propagation and processing delays, while $\wms$ and $\wsm$ model the random queuing delays.

Assuming the values of $\delta$, $\dms^{\min}$ and $\dsm^{\min}$ remain constant over the duration of $P$ two-way message exchanges, we can collect multiple observation pairs $(y_1,y_2)$ to help estimate $\delta$. 
We denote these observations as
\begin{align}
y^*_{i,1} = \dms^{\min}+ \delta + w_{i,1} \  ,\quad   y^*_{i,2} = \dsm^{\min} - \delta + w_{i,2}   \label{eq:Ptimestampdiffs}
\end{align}
for $i=1,\cdots,P$. We now consider three observation models:
\begin{enumerate}
\item \textit{Known fixed delay model (K-model)}: 
Here we assume that $\dms^{\min}$ and $\dsm^{\min}$ are fully known at the slave. 
Hence, setting ${y_{i,k}= y^*_{i,k}-d_{j}^{\tm{min}}}$, we obtain the compensated observations
\begin{align}
y_{i,1} = \delta + w_{i,1} \  ,\quad   y_{i,2} = -\delta + w_{i,2}   
\end{align}
for $i=1,\cdots,P$. These observations can be collected to obtain the vector observation model
\begin{align}
\vct{y} =  \delta \vct{e} + \vct{w} \label{eq:knownfixeddelay_obsmodel_vct}
\end{align}
where 
\begin{align}
&\vct{y} = \begin{bmatrix} \vct{y}_1^\tr \ \vct{y}_2^\tr \end{bmatrix}^\tr,\ \vct{y}_k = \left[y_{1,k} \ \cdots \ y_{P,k}\right]^\tra \label{eq:yvct}\\
&\vct{w} = \begin{bmatrix} \vct{w}_1^\tr \  \vct{w}_2^\tr \end{bmatrix}^\tr, \ 
\vct{w}_k = \left[w_{1,k} \ \cdots \ w_{P,k} \right]^\tra  \label{eq:wvct}\\
&\vct{e} = [\vct{1}_P \ \ \  (-\vct{1}_P) ]^\tra 
\end{align}
and $\vct{1}_N$ is a $N\times 1$ vector with all elements equal to $1$.

\item \ti{Standard model (S-model)}:
Here we assume that only the difference between $\dms^{\min}$ and $\dsm^{\min}$, referred to as the \textit{delay asymmetry}, is known to the slave. By compensating the observations as  
\begin{align}
y_{i,1}= y^*_{i,1}, \ \ y_{i,2}= y^*_{i,2} - \dsm^{\min} + \dms^{\min}
\end{align}
we obtain 
\begin{align}
y_{i,1} = \dmin + \delta + w_{i,1} \  ,\quad   y_{i,2} = d - \delta + w_{i,2}   \label{eq:stdmodelobs}
\end{align}
for $i=1,\cdots,P$, where $\dms^{\min}=d$. 
These observations can be denoted vectorially as 
\begin{align}
\vct{y} =  d\vct{1}_{2P} + \delta \vct{e} + \vct{w}  = \mat{A} \vct{\theta} + \vct{w}\label{eq:standard_obsmodel_vct}
\end{align}
where $\vct{y}$ and $\vct{w}$ are as defined in (\ref{eq:yvct}) and (\ref{eq:wvct}), and 
\begin{align}
& \vct{\theta} = [\theta_1\ \theta_2]^\tra = [d+\delta\ \ d-\delta]^\tra , \\
&\mat{A} = 
\left[ 
\begin{array}{cc}
\mat{1}_{P} & \mat{0}_{P} \\
\mat{0}_{P} & \mat{1}_{P} 
\end{array}
\right] , \label{eq:smodelvct}
\end{align}
with $\mat{1}_{Q}$, $\mat{0}_{Q}$ represent $Q\times 1$ vectors of ones and zeros, respectively.

Note that this model also covers the case of symmetric path delays, where $\dms^{\min} = \dsm^{\min}$, and hence the delay asymmetry is zero.
We further note that other cases where the relationship between the fixed delays is known, such as the case where the ratio $\dms^{\min}/\dsm^{\min}$ is known, can also be handled using a model similar to (\ref{eq:standard_obsmodel_vct}). For brevity, only the case of known delay asymmetry is considered here.

\item \ti{Multiblock model (M-model)}: Here we assume, as in the standard model, that the delay asymmetry is known to the slave.
Suppose we refer to a set of $P$ observation pairs as a \textit{block}. In this model, we further assume that in addition to the current block, we have observation pairs from $B$ previous blocks available to us. 
The phase offset $\delta$ is modeled as being constant for all observation pairs within each block, but varying between different blocks. The fixed delay $d$ is modeled as constant across all $B+1$ blocks. This model is representative of scenarios where changes in the fixed delay occurs over longer time scales than changes in phase offset. We denote observations pairs in past blocks using the notation
\begin{align}
y_{i,j,1} = \dmin + \delta_j + w_{i,j,1} \  ,\quad   y_{i,j,2} = d - \delta_j + w_{i,j,2}   
\end{align}
and observation pairs in the current block as
\begin{align}
y_{i,1} = \dmin + \delta + w_{i,1} \  ,\quad   y_{i,2} = d - \delta + w_{i,2}   
\end{align}
for $i=1,\cdots,P$ and $j=1,\cdots,B$.

We thus obtain the vector observation model
\begin{align}
 \vct{y} = \mat{G}\vct{\theta} + \vct{w}  \label{eq:multiblock_obsmodel_vct}
\end{align}
where 
\begin{align}
&\vct{y} = \begin{bmatrix} \vct{y}_1^\tr \ \vct{y}_2^\tr \end{bmatrix}^\tr,\\
&\vct{y}_k = \left[y_{1,k}\ \cdots\ y_{P,k} \ y_{1,1,k}\ y_{1,2,k}\ \cdots \ y_{B,P,k}\right]^\tra \label{eq:yvct1}\\
&\vct{w} = \begin{bmatrix} \vct{w}_1^\tr \  \vct{w}_2^\tr \end{bmatrix}^\tr,\\ 
&\vct{w}_k = \left[w_{1,k}\ \cdots\ w_{P,k} \ w_{1,1,k}\ w_{1,2,k}\ \cdots \ w_{B,P,k}\right]^\tra \label{eq:wvct1}\\
&\vct{\theta} = [d\  \delta\ \  \delta_{1}\ \ \cdots \ \  \delta_{B}]^\tra , \\
&\mat{G} =  \left[ \vct{1}_{2BP} \ \ \ \mat{Z} \otimes \vct{1}_{P} \right], \quad  \mat{Z} = [\mat{I}_B\ \ \  (-\mat{I}_B) ]^\tra 
\end{align}
and $\mat{I}_B$, $\otimes$ denote the identity matrix of size $B$ and the Kronecker product operator, respectively.

\end{enumerate}

In practice, the S-model and M-model are more practical than the K-model, since they only assume that the difference between the fixed delays is known. 
We still consider the K-model in this paper since it provides us with useful limits on the performance of optimum estimators under the M-model.

Given either of the observation models, the problem of POE is to estimate $\delta$ from the observation vector $\vct{y}$. 
Here we further make the following assumptions:
\begin{enumerate}[(i)]
\item All the queuing delays are strictly positive random variables that are mutually independent.
\item All forward queuing delays share a common pdf $f_1(w)$. Similarly the reverse queuing delays share a common pdf $f_2(w)$.
\item The maximum possible value for a forward or reverse queuing delay is finite.
\item All the unknown fixed delays and phase offsets are deterministic parameters, i.e. no probability distributions for these parameters are known a priori.
\end{enumerate}

\section{Minimax Estimation for General Location Parameter Problems} \label{sec:minimaxest}
We now consider a general class of estimation problems, where the effect of the unknown parameters is to shift the location of the pdf of the observations without modifying the underlying shape of the pdf.
The POE problems under all three observation models considered in Section \ref{sec:PDV_SysModel}
belong to this general class of problems.
The general results derived here shall be applied to the POE models in Section \ref{sec:poe_simplifications}. The proof of all the lemmas and theorems stated in this section are provided in the appendix.

We first define the general class of problems we are interested in studying.
\begin{definition}[Vector Location Parameter Problem]\label{defn:locparamprob}
Suppose we want to estimate a linear combination $\vct{c}^\tr \vct{\theta}$ of the unknown parameters contained in $\vct{\theta} \in \mathbb{R}^M$ (where $\vct{c}\in \mathbb{R}^M$ is a constant vector), based on observations $\vct{x}\in \mathbb{R}^N$. 
If the observations have a pdf of the form
\begin{align}
f(\vct{x} |\vct{\theta} ) = f_0(\vct{x} - \mat{G}\vct{\theta}) \label{eq:loc_param_prob_defn}
\end{align}
for some $N\times M$ matrix $\mat{G}$ and function $f_0(\cdot)$,
then we shall refer to such an estimation problem as a \textit{vector location parameter problem}.
\end{definition}
All the definitions and theorems in the remainder of this section apply specifically to this vector location parameter problem.
The results we derive further require that the function $f_0(\vct{x})$ be non-zero over a bounded, positive range of values of its arguments, as defined below.
\begin{definition}[Finite Support]\label{defn:finitesupport}
We say that $f_0(\vct{x})$ in (\ref{eq:loc_param_prob_defn}) has \textit{finite support} if there exists a finite $L > 0$ such that $f_0(\vct{x}) = 0$ whenever all the elements of the vector $\vct{x}$ lie outside the interval $[0,L]$.
\end{definition}

It is typical in statistical literature to characterize the performance of an estimator via the mean squared error (MSE) metric.  
There are three ways to define the MSE metric:
\begin{enumerate}
\item The \textit{conditional MSE} 
\begin{align}
\mathcal{R}(g(\vct{x}),\vct{\theta}) &=  \int_{\mathbb{R}^N}  [g(\vct{x})-\vct{c}^\tr \vct{\theta}]^2 f(\vct{x} |\vct{\theta} ) \tm{d}\vct{x}
\end{align}
\item The \textit{maximum MSE}
\begin{align}
\mathcal{M}(g(\vct{x})) = \sup_{\vct{\theta} \in \mathbb{R}^M} \mathcal{R}(g(\vct{x}),\vct{\theta})
\end{align}
\item The \textit{average MSE} 
\begin{align}
\mathcal{B}(g(\vct{x}),p(\vct{\theta})) = \int_{\mathbb{R}^M} \mathcal{R}(g(\vct{x}),\vct{\theta})p(\vct{\theta}) \tm{d}\vct{\theta}
\end{align}
where $p(\vct{\theta})$ is a prior distribution defined over $\vct{\theta} \in \mathbb{R}^M$.
\end{enumerate}
In this section, we consider the problem of finding estimators that are optimum in terms of minimizing the maximum MSE, and refer to such estimators as \textit{minimax} estimators. 
The definitions of the conditional and average MSEs shall be used in the proofs of the optimality of the minimax estimator.

We now consider a class of estimators known as shift invariant estimators, defined as follows.
\begin{definition}[Shift Invariant Estimator]\label{definition:ShiftInvariance}
	We say that an estimator $g(\vct{x})$ of $\vct{c}^\tr \vct{\theta}$ is \textit{shift invariant} if for the same matrix $\mat{G}$ used in (\ref{eq:loc_param_prob_defn}),
	\begin{align}
	g(\vct{x}+ \mat{G}\vct{h}) = g(\vct{x}) + \vct{c}^\tr \vct{h}  \label{eq:shift_invar_est_defn}
	\end{align}
	for all $\vct{h} \in \mathbb{R}^{M}$.
\end{definition}
While the conditional, maximum and average MSEs are in general different for a estimator, for a shift invariant estimator they are always equal, as stated in the following lemma.
\begin{lemma}\label{lemma:shiftinv_lemma}
Any shift invariant estimator $g(\vct{x})$ of $\vct{c}^\tr \vct{\theta}$  has a conditional MSE that is constant with respect to $\vct{\theta}$, and satisfies
\begin{align}
\mathcal{R}(g(\vct{x}),\vct{\theta}) = \mathcal{M}(g(\vct{x}))= \mathcal{B}(g(\vct{x}),p(\vct{\theta}))
\end{align}
for any choice of prior distribution $p(\vct{\theta})$.
\end{lemma}

We now give the expression for the minimax estimator and prove its optimality using Definition \ref{definition:ShiftInvariance} and Lemma \ref{lemma:shiftinv_lemma}.
\begin{theorem}[Minimax estimator]\label{theorem:MinimaxEstimator}
If $f_0(\vct{x})$ has finite support, then the estimator 
\begin{align}
g^*(\vct{x}) = 
\frac{\int_{\mathbb{R}^{M}}  [\vct{c}^\tra\hat{\vct{\theta}}]  f(\vct{x}| \hat{\vct{\theta}}) \ \tm{d}\hat{\vct{\theta}}}
{\int_{ \mathbb{R}^{M}}   f(\vct{x} | \hat{\vct{\theta}}) \ \tm{d}\hat{\vct{\theta}}} \label{eq:minimaxoptestimator}
\end{align}
satisfies the following properties:
\begin{enumerate}[(i)]
\item $g^*(\vct{x})$ is shift invariant. 
\item $g^*(\vct{x})$ is a minimax estimate of $\vct{c}^\tra \vct{\theta}$.
\item Among all estimators of $\vct{c}^\tra \vct{\theta}$ that are shift invariant, $g^*(\vct{x})$ achieves the minimum conditional MSE $\mathcal{R}(g(\vct{x}),\vct{\theta})$.
\item $g^*(\vct{x})$ is unbiased, i.e. $\Ebr{\left[ g^*(\vct{x}) - \vct{c}^\tra \vct{\theta}  \right] \ |\ \vct{\theta}}= 0$.
\end{enumerate}
\end{theorem}

An interesting property of the minimax estimator is that for a given set of observations, the minimax estimate of a linear combination of parameters is identical to the same linear combination of the minimax estimates of each of the parameters. 
Formally, this can be stated as follows.
\begin{lemma} \label{lemma:linearsummation}
Let $\vct{\theta}=[\theta_1\ \cdots\ \theta_M]^T$, and let $g_i^*(\vct{x})$ represent the minimax estimate of $\theta_i$. 
If $\vct{c} = [c_1\ \cdots\ c_M]^\tra$, then the minimax estimate $g^*(\vct{x})$ of $\vct{c}^\tra \vct{\theta}$ satisfies
\begin{align*}
g^*(\vct{x}) = \sum_{i=1}^{M} c_i g_i^*(\vct{x})
\end{align*}
\end{lemma}
This property will allow us to simplify the form of the minimax estimator under the S-model in Section \ref{sec:poe_simplifications}.
Another interesting property of the minimax estimator emerges when we consider multiple minimax estimates, each based on a different observation vector. 
Here we can show that the sum of the MSEs of the individual minimax estimates will always be  less than the MSE of the minimax estimate based on sum of all the observation vectors. This can be formally stated as follows.
\begin{theorem}\label{theorem:maxmselinearity}
Let $\vct{x}_1,\cdots,\vct{x}_K$ be $N$-dimensional random vectors with pdfs of the form
\begin{align}
f(\vct{x}_k|\vct{\theta}_k)= f_k(\vct{x}_k - \mat{G}_k\vct{\theta}_k)
\end{align}
where $f_k(\cdot)$ has finite support for $k=1,\cdots,K$. Assume that $\vct{x}_1,\cdots,\vct{x}_K$ are all mutually independent conditioned on the unknown parameters, i.e. the joint pdf $f(\vct{x}_{k_1},\vct{x}_{k_2}|\vct{\theta}_{k_1},\vct{\theta}_{k_2})$ satisfies
\begin{align}
f(\vct{x}_{k_1},\vct{x}_{k_2}|\vct{\theta}_{k_1},\vct{\theta}_{k_2}) = f(\vct{x}_{k_1}|\vct{\theta}_{k_1}) f(\vct{x}_{k_2}|\vct{\theta}_{k_2}) \label{eq:indepdenceassumption}
\end{align} 
for all values of $k_1$ and $k_2$.
Let $h_k^*(\vct{x}_k)$ denote the minimax estimate of $\vct{c}^\tra \vct{\theta}_k$.
Further, let $\vct{x} =\sum_{k=1}^{K} \vct{x}_k$, $\vct{\theta} = \sum_{k=1}^{K}\vct{\theta}_k$, and let $g^*(\vct{x})$ denote the minimax estimate of $\vct{c}^\tra \vct{\theta}$ from $\vct{x}$.
Then $g^*(\vct{x})$ satisfies
\begin{align}
\mathcal{M}(g^*(\vct{x})) \ge \sum_{k=1}^{K}\mathcal{M}(h^*_k(\vct{x}_k))
\end{align}
\end{theorem}
This property will be useful in proving certain properties of the minimax estimator for POE in Section \ref{sec:minimax_cross_traf}.

\section{Simplification of Minimax Estimator for the POE problem}\label{sec:poe_simplifications}
We now use the results in Section \ref{sec:minimaxest} to obtain minimax optimum estimators under the three POE observation models discussed in Section \ref{sec:PDV_SysModel}, and simplify the resulting expressions. 
\begin{enumerate}[1)]
\item \textit{Known fixed delay model}: 
As stated in (\ref{eq:knownfixeddelay_obsmodel_vct}), the pdf of the observation vector $\vct{y}$ has the form
\begin{align}
f(\vct{y} |\delta  ) &= f_{\vct{W}}(\vct{y} - \delta \vct{e}) 
\end{align}
where
\begin{align}
f_{\vct{W}}(\vct{w}) &= \prod_{i=1}^{P} f_1(w_{i,1}) f_2(w_{i,2})
\end{align}
Hence, according to Definition \ref{defn:locparamprob}, this is a vector location parameter problem. 
Thus, using Theorem \ref{theorem:MinimaxEstimator}, we obtain the minimax estimator of $\delta$ as 
\begin{align}
\hat{\delta } (\vct{y})= \frac{\int_{\mathbb{R}} \delta f_{\vct{W}}(\vct{y} - \delta \vct{e}) \tm{d}\delta }{\int_{\mathbb{R}} f_{\vct{W}}(\vct{y} - \delta \vct{e}) \tm{d}\delta}\label{eq:kmodmm}
\end{align}

\item \textit{Standard Model}:
As stated in (\ref{eq:standard_obsmodel_vct}), here the pdf of the observation vector $\vct{y}$ has the form
\begin{align}
f(\vct{y} |\vct{\theta} ) &= f_{\vct{W}}(\vct{y} - \mat{A} \vct{\theta}) \\
&= f_{\vct{W},1}(\vct{y}_1 - \theta_1\mat{1}_{P} ) 
f_{\vct{W},2}(\vct{y}_2 - \theta_2\mat{1}_{P} )
\end{align}
where
\begin{align}
f_{\vct{W}}(\vct{w}) &= \prod_{i=1}^{P} f_1(w_{i,1}) f_2(w_{i,2})\ ,\\
f_{\vct{W},k}(\vct{w}_k) &= \prod_{i=1}^{P} f_k(w_{i,k}) \quad \tm{for}\ \ k=1,2 
\end{align}
Hence, according to Definition \ref{defn:locparamprob}, this is a vector location parameter problem. 
Our goal is to estimate ${\delta = \vct{c}^\tra \vct{\theta}}$ (where $\vct{c} = [0.5\ -0.5]^\tra$) from the observation vector  $\vct{y}= \mat{A} \vct{\theta} + \vct{w}$. 
Hence, using Theorem \ref{theorem:MinimaxEstimator}, we obtain the minimax estimate 
\begin{align}
& \hat{\delta}(\vct{y}) = 
\frac{\int_{ \mathbb{R}^{2}}  [\vct{c}^\tra\vct{\theta} ] f(\vct{y}| \vct{\theta}) \ \tm{d}\vct{\theta}}
{\int_{ \mathbb{R}^{2}}   f(\vct{y}| \vct{\theta}) \ \tm{d}\vct{\theta}} 
\end{align}
Using Lemma \ref{lemma:linearsummation}, we can further simplify the estimator as

\begin{align}
\hat{\delta}(\vct{y}) &= 
\frac{1}{2}
\Bigg[ 
\frac{
\int_{\mathbb{R}} \theta_1 f_{\vct{W},1}(\vct{y}_1 -\theta_1  \vct{1}_{P})  \ \tm{d}\theta_1 
}
{
\int_{\mathbb{R}}   f_{\vct{W},1}(\vct{y}_1 -\theta_1  \vct{1}_{P}) 
\ \tm{d}\theta_1 
} \nn \\  
&\qquad \qquad -  
\frac{
\int_{\mathbb{R}} \theta_2 f_{\vct{W},2}(\vct{y}_2 -\theta_2  \vct{1}_{P})  \ \tm{d}\theta_2 
}
{
\int_{\mathbb{R}}   f_{\vct{W},2}(\vct{y}_2 -\theta_2  \vct{1}_{P}) 
\ \tm{d}\theta_2 
}  
\Bigg] \label{eq:smodmm}
\end{align}

\item \textit{Multiblock Model}: 
As stated in (\ref{eq:multiblock_obsmodel_vct}), here the pdf of the observation vector $\vct{y}$ has the form
\begin{align}
f(\vct{y} |\vct{\theta} ) &= f_{\vct{W}}(\vct{y} - \mat{G} \vct{\theta}) 
\end{align}
where
\begin{align}
f_{\mat{W}}(\vct{w}) &=\prod_{i=1}^{B}\prod_{j=1}^{P}  \prod_{k=1}^{2} f_{k}(w_{i,j,k})
\end{align}
Hence, according to Definition \ref{defn:locparamprob}, this is also a vector location parameter problem.
Our goal is to estimate $\delta = \hat{\vct{c}}^\tra \vct{\theta}$ from $\vct{y}$, where
$\hat{\vct{c}} = [0\ 1\ \underbrace{0\ \cdots\ 0}_{B-1\ \tm{zeros}}]^\tra$. 
Using Theorem \ref{theorem:MinimaxEstimator}, we obtain the minimax estimate
\begin{align}
\hat{\delta}(\vct{y}) = 
\frac{\int_{\mathbb{R}} \delta \Gamma(\delta, \vct{y}) \tm{d}\delta}
{\int_{\mathbb{R}}  \Gamma(\delta, \vct{y}) \tm{d}\delta} \label{eq:mmodmm}
\end{align}
where 
\begin{align}
\Gamma(\delta,\vct{y}) &= \int_{\mathbb{R}} \left[ \prod_{i=1}^{P} \prod_{k=1}^{2} f_k(y_{i,k} - d + (-1)^k \delta ) \right] \nn \\
&\qquad\qquad\qquad\qquad \cdot   \Omega(d,\vct{y})\ \tm{d}(d) \label{eq:mmgam}\\
 \Omega_j(d,\vct{y}) &= \prod_{j=1}^{B} \left[\int_{\mathbb{R}} \prod_{i=1}^{P} \prod_{k=1}^{2}  f_k(y_{i,j,k} - d + (-1)^k \delta_j ) \tm{d} \delta_{j} \right] \label{eq:mmomega}
\end{align}
\end{enumerate}

In scenarios where analytical expressions for the queuing delay pdfs $\fms{w}$ and $\fsm{w}$ are known, it might be possible to further simplify the integrals in (\ref{eq:kmodmm}), (\ref{eq:smodmm}) and (\ref{eq:mmodmm})-(\ref{eq:mmomega}). 
In the more general case of arbitrary pdfs $\fms{w}$ and $\fsm{w}$, these integrals can be computed by approximating them with Riemann summations. In such cases, the computational complexity associated with the minimax estimators will depend on the number of bins used in the Riemann summations. 
Typically, this computational complexity is significantly higher than that of conventional estimators such as the sample minimum, mean, median or maximum estimators.

Due to the nature of the POE observation models, some comments on the MSE of the minimax optimum estimator (or the \textit{minimax MSE}) can be made directly, without requiring numerical evaluations. 
Firstly, the minimax MSE under the K-model is guaranteed to be lower than that under the S-model or M-model, since the nuisance parameter $d$ is absent from the K-model. 
Further, the minimax MSE under the M-model is guaranteed to be lower than that under the S-model, since the M-model has additional information from past blocks available to it. 
This past information can be used to reduce the uncertainty associated with the nuisance parameter $d$, and hence improve the estimate of $\delta$.

\section{Minimax MSE under IID single-node queuing delays}\label{sec:minimax_cross_traf}
The performance of the minimax estimators described in Section \ref{sec:poe_simplifications} depends on the nature of the network queuing delays, which in turn depends on the number of nodes present between the master and the slave. 
Theorem \ref{theorem:maxmselinearity} can be used to obtain a simple relationship between the minimax MSE and the number of intermediate nodes, under certain network conditions.
We state this relationship in the form of the following corollary to Theorem \ref{theorem:maxmselinearity}, with the proof provided in the appendix.
\begin{corollary}\label{theorem:msevsswitches}
Consider a network consisting of a master and a slave separated by $N$ nodes. 
Let $\minimaxmse(N)$ represent the minimax MSE associated with POE under the S-model in this scenario, for a fixed number of two-way message exchanges.
Let the \textit{single-node queuing delay} refer to the queuing delay experienced by packets at any single node\footnote{Measurements of the single node queuing delay in the forward or reverse direction would corresponding to the proper entries of the vector $\vct{w}$ in (\ref{eq:standard_obsmodel_vct}) for the case where only N=1 node is involved.}. 
Assume that the single-node queuing delays across all nodes in the forward direction are independent and identically distributed (i.i.d.). Assume that the same is true in the reverse direction as well. 
Then $\minimaxmse(N)$ satisfies
\begin{align}
\minimaxmse(KL) \ge K \minimaxmse(L) \label{eq:smodelmselinearity}
\end{align}
where $K$ and $L$ are any two positive integers.
\end{corollary}
For $L=1$, the relation in (\ref{eq:smodelmselinearity}) reduces to $\minimaxmse(K) \ge K \minimaxmse(1)$, which essentially implies that \textit{in networks with i.i.d. single-node queuing delays at all intermediate network nodes, the minimax MSE grows at least linearly with the number of nodes}. 
This interpretation can be especially useful for network designers, since it provides a computationally simple upper limit on the number of nodes that can be allowed between the master and the slave for a given synchronization accuracy requirement.  
A typical example where independent, identically distributed single-node queuing delay distributions can be assumed is a network in which only \textit{cross traffic flows} (defined in Section \ref{sec:simresults}) are present. 
Note that a relationship similar to (\ref{eq:smodelmselinearity}) can also be derived under the K-model and the M-model. For brevity, only the S-model is considered in this section.

\section{Simulation Results}\label{sec:simresults}
We now compare the performance of conventional POE schemes against the newly derived minimax estimators. 
To this end, we consider a few network scenarios motivated by the ITU-T recommendation G.8261 \cite{g8261}. 
The metric we use to quantify estimator performance is the maximum MSE. \
For brevity, we refer to the maximum MSE as simply the MSE throughout this section.

We consider four commonly used conventional POE schemes, namely the sample minimum, maximum, mean and median filtering schemes. Given the observation vector $\vct{y}$ of either the K-model or the S-model, these schemes use an estimator of the form
\begin{align}
\hat{\delta} = \frac{\convest(\vct{y}_1) - \convest(\vct{y}_2)}{2} \label{eq:convest}
\end{align}
where $\convest(\vct{x})$ denotes either the minimum, maximum, mean or median of the elements of the vector $\vct{x}$. 
Under the M-model, these estimators behave exactly as under the S-model, discarding information from past blocks since they have no means of utilizing it. 
It is easy to show that these estimators are shift invariant under all three observation models. They also have an identical value for the MSE across all three models, given as
\begin{align}
\mathcal{M}(\hat{\delta}) &= \Ebr{\hat{\delta}^2\ \Big|\ \vct{\theta}=\left[ \begin{array}{c}0\\ 0\end{array}\right]} = \sigma^2 + \mu^2 \label{eq:conv_est_mse_breakup}
\end{align}
where 
\begin{align}
\sigma^2 &= \tm{var}\left\{\hat{\delta}^2\ \Big|\ \vct{\theta}=\left[ \begin{array}{c}0\\ 0\end{array}\right]\right\} \\
&= \frac{1}{4}\big[ \var{\convest(\vct{w}_1)} + \var{\convest(\vct{w}_2)} \big] \label{eq:conv_est_var}
\end{align}
represents the estimator variance, while 
\begin{align}
\mu =\frac{1}{2}\Big[\E{\convest(\vct{w}_1)} - \E{\convest(\vct{w}_2)}\Big] \label{eq:conv_est_bias}
\end{align}
represents the estimator bias. Note that
\begin{align}
E[g(\vct{w}_i)] &= \int  \convest(\vct{w}_{i}) f_{\vct{W}_i}(\vct{w}_i) \tm{d}\vct{w}_{i} \ , \\
\var{\convest(\vct{w}_i)} &= \int_{\vct{w}_i} \left\{\convest(\vct{w}_i)-E[\convest(\vct{w}_i)]\right\}^2 f_{\vct{W}_i}(\vct{w}_i) \tm{d}\vct{w}_i\ , \\
f_{\vct{W}_i}(\vct{w}_i) &= \prod_{j=1}^{P}  f_{i}(w_{i,j})
\end{align}
It is easy to see from (\ref{eq:conv_est_bias}) that when the forward and reverse queuing delay distributions $f_1(w)$ and $f_2(w)$ are not identical, $\mu$ can be non-zero, and hence have a significant contribution in the MSE expression in (\ref{eq:conv_est_mse_breakup}). 
This can be avoided by subtracting out the bias, to obtain the unbiased estimate
\begin{align}
\tilde{\delta} = \hat{\delta} - \mu \label{eq:convestunbiased}
\end{align}
Hence, in our results, we measure the performance of conventional estimators as their MSE after their bias has been compensated.
 
In order to obtain the queuing delay distributions, we consider a Gigabit ethernet network consisting a cascade of $20$ switches between the master and slave nodes. Each switch is assumed to be a store-and-forward switch, which implements strict priority queuing. 
We consider two types of background traffic flows in this network:
\begin{enumerate}
	\item  \textit{Cross traffic flows}: In such traffic flows \cite{g8261}\cite{guru2014perf}, fresh background traffic packets are injected at each node along the master-slave path, and these packets exit the master-slave path at the subsequent node (see 4-switch example in Fig. \ref{fig:Fourswitch_cross}). 
	The arrival times and sizes of the packets injected at each switch are assumed to be statistically independent of that of packets injected at other switches.
	\item \textit{Mixed traffic flows}: Here a mixture of cross traffic flows and \textit{inline traffic flows} are present in the network. Inline traffic flows \cite{anyaegbu2012sample} are characterized by packets that are injected only at the first switch along the master slave path, and that travel along the same path as synchronization traffic through the entire cascade of switches (see 4-switch example in Fig. \ref{fig:Fourswitch_inline}).
\end{enumerate}
\begin{figure}
\centering
\subfloat[Cross traffic flows] {\includegraphics[width=0.8\columnwidth]{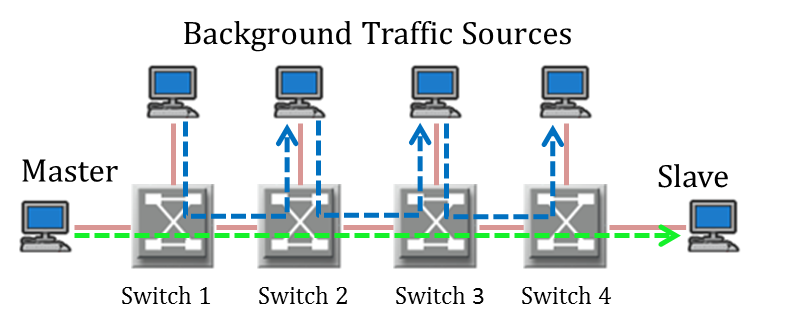}\label{fig:Fourswitch_cross}}\\
\subfloat[Inline traffic flows] {\includegraphics[width=0.8\columnwidth]{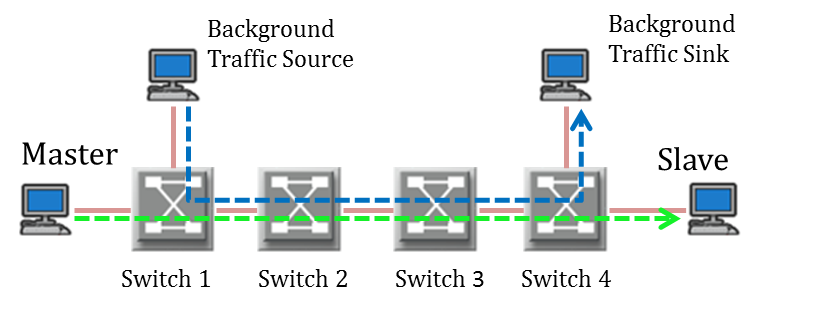}\label{fig:Fourswitch_inline}}
\caption{Examples of four switch networks with cross and inline traffic flows. Red lines indicate network links, blue lines indicate the direction of background traffic flows, and green line represents the direction of synchronization traffic flows.}
\end{figure}
With regard to the distribution of packet sizes in background traffic, we consider Traffic Models 1 (TM1) and 2 (TM2) from the ITU-T recommendation G.8261 \cite{g8261} for cross traffic flows, as specified in Table \ref{table:bgtrafmodels}. 
For inline traffic flows, we consider a third traffic model where packet sizes are uniformly distributed between 64 and 1500 bytes [7].
\begin{table}
\centering
\begin{tabular}{|c|c|c|}
\hline 
\textbf{Traf. Model} & \textbf{Packet Sizes (Bytes)} & \textbf{\% of Load}\\ 
\hline 
TM1 &  $\{64 , 576 , 1518\}$ & $\{80\%, 5\%, 15\%\}$\\
\hline 
TM2 &  $\{64 , 576 , 1518\}$ & $\{30\%, 10\%, 60\%\}$\\ \hline 
\end{tabular}
\caption{Models for composition of background traffic packets}
\label{table:bgtrafmodels}
\end{table}
We assume that the interarrival times between packets in all background traffic flows follow exponential distributions.
We refer to the percentage of the link capacity consumed by background traffic as the \textit{load}. 
In order to achieve a particular load, we accordingly set the rate parameter of each exponential distribution. 
The queuing delay distributions under a number of network scenarios are plotted in Fig. \ref{fig:pdfs}.
These distributions were obtained empirically using low-level queue simulations. 
Without loss of generality, we assume that fixed delay components of the ETE delays equal
zero, hence the support of the queuing delay distributions always begins at zero in the plots.

The MSE of various estimators under different observation models and network conditions are plotted versus the number of observation pairs/samples $P$  in Figs. \ref{fig:mse_symm_cross} - \ref{fig:mse_asymm_mixed}. 
In order to compute the minimax estimates, the integrals in (\ref{eq:kmodmm}), (\ref{eq:smodmm}) and (\ref{eq:mmodmm}) were replaced with Riemann sums. The spacing between adjacent Riemann summation bins was set to $0.001$ $\mu$s, to ensure that the additional error introduced due to the Riemann sum approximation is small relative to the MSE being computed.
Further, to facilitate comparisons against the LTE synchronization requirement of $1.25$ $\mu$s of synchronization accuracy, the estimation error standard deviation required so that the absolute estimation error lies
under $1.25$ $\mu$s  with a $5$-sigma level of certainty is also plotted over the curves. Here the $5$-sigma level of certainty implies that on average, only about $6$ out of $10^6$ estimates will
have absolute estimation error that exceeds  $1.25$ $\mu$s.
Some key observations we can make from the results are:
\begin{enumerate}[1.]
\item \textit{Performance under symmetric cross traffic (Fig. \ref{fig:mse_symm_cross})}: Here, the gap between the K-model and S-model minimax estimators is negligible under all four loads ($20\%$, $40\%$, $60\%$, $80\%$) considered. 
Hence, under these network scenarios, there is little performance to be gained from the additional knowledge about fixed delays that the K-model provides over the S-model. 
Further, while the sample minimum estimator performs near-optimally at $20$\% load, at higher loads none of the conventional estimation schemes come close to achieving minimax optimal performance. 
In fact, at $80\%$ load, the minimax estimators achieve the LTE synchronization requirement using only about $200$ samples, while about $800$ samples are required by the best conventional estimator. 
\item \textit{Performance under symmetric mixed traffic (Fig. \ref{fig:mse_symm_mixed})}: Here, there is a fair gap between the K-model and S-model minimax estimators under the lower load scenario of Fig. \ref{fig:mse_symm_mixed_low}, which disappears under the higher load scenario of Fig. \ref{fig:mse_symm_mixed_high}. Further, the S-model minimax estimator requires about $50\%$ fewer samples than the best conventional estimator, to achieve the LTE synchronization requirement under the low load scenario.
Interestingly, the sample mean filter performs near-optimally under the high load scenario.
This indicates that the performance gap between the best conventional estimator and the minimax estimator may need to be studied on a per-case basis, and predicting general trends might be difficult.
\item \textit{Performance under asymmetric traffic (Figs. \ref{fig:mse_asymm_cross} and \ref{fig:mse_asymm_mixed})}: Here there is a significant gap between the K-model and S-model minimax estimators, with the K-model minimax estimator requiring about $90\%$ and $22\%$ fewer samples than the S-model estimator in  Fig. \ref{fig:mse_asymm_cross} and Fig. \ref{fig:mse_asymm_mixed}, respectively, in order to meet the LTE synchronization requirement threshold. 
This is expected in cases where the queuing delay distribution in one network direction has significantly lower spread than in the other direction. 
In such cases, the MSE of conventional estimators,
given by (\ref{eq:conv_est_var}), is dominated by either the
first or second term in (\ref{eq:conv_est_var}) if one these variances
is much larger than the other. 
On the other hand, the K-model estimator can utilize knowledge of the fixed delays to base its estimate on only the observations corresponding to the direction with lower variance, thereby eliminating large contributions to its MSE caused by the queuing delay distribution that has higher variance. 

Further, since the M-model estimator can use information from $B$ past blocks to estimate the fixed delay, we expect it to achieve the performance of the K-model estimator in the limiting case where $B\rightarrow \infty$.
In our simulations, we observe that M-model minimax estimator closely approaches the K-model minimax estimator in performance for  fairly small values of $B$ (between $5$ and $20$).
\end{enumerate}

\section{Conclusions}
We derived minimax optimum estimators for a general class of location parameter problems, and applied them to the problem of phase offset estimation under multiple observation models.
In cases where the pdf of the queuing delays are known, minimax estimators can be used to obtain the best possible estimation performance. 
The MSE curves of the minimax estimators can also serve as a design tool for practical synchronization deployments, by providing fundamental limits on POE performance for a given set of network conditions.
Our simulation results indicate that conventional estimators can perform close to optimum in certain low-load scenarios with symmetric queuing delay distributions.
However, optimum estimators appear to provide significant performance benefits in scenarios where the queuing delay distributions are asymmetric, a case that occurs frequently in practice.
The results in this paper could help guide the development of new POE schemes that address synchronization challenges arising in current and future generations of mobile networks.

\appendix
\begin{proof}[Proof of Lemma \ref{lemma:shiftinv_lemma}]
For any shift invariant estimator $g(\vct{x})$, we can show that if $\vct{\theta}_1$ and $\vct{\theta}_2$ are any two values of the parameter vector with $\vct{h} = \vct{\theta}_1 - \vct{\theta}_2$, then
\begin{align}
&\mathcal{R}(g(\vct{x}),\vct{\theta}_1) \nn \\
&= \int_{\mathbb{R}^{N}}  [g(\vct{x})-\vct{c}^\tr (\vct{\theta}_2+\vct{h})]^2 f_0(\vct{x} -\mat{G}(\vct{\theta}_2+\vct{h}) ) \tm{d}\vct{x}\\
&= \int_{\mathbb{R}^{N}}  [g(\vct{x}-\mat{G}\vct{h})-\vct{c}^\tr \vct{\theta}_2]^2 f((\vct{x}-\mat{G}\vct{h}) |\vct{\theta}_2 ) \tm{d}\vct{x}\\
&= \int_{\mathbb{R}^{N}}  [g(\vct{x})-\vct{c}^\tr \vct{\theta}_2]^2 f(\vct{x} |\vct{\theta}_2 ) \tm{d}\vct{x}\\
&\quad\quad  \tm{(using a change of variables)} \nn \\
&= \mathcal{R}(g(\vct{x}),\vct{\theta}_2) 
\end{align}
Hence $g(\vct{x})$ has constant conditional MSE w.r.t. $\vct{\theta}$. Further, using the definitions of the maximum and average MSEs, we obtain $\mathcal{R}(g(\vct{x}),\vct{\theta}) = \mathcal{M}(g(\vct{x}))= \mathcal{B}(g(\vct{x}),p(\vct{\theta}))$.
\end{proof}
$ $\\

\begin{proof} [Proof of Theorem \ref{theorem:MinimaxEstimator}] \label{Appendix:MinimaxProof}
\begin{enumerate}[(i)]
\item It is simple to show that $g^*(\vct{x})$ is shift invariant, since
\begin{align}
g^*(\vct{x} + \mat{G}\vct{h}) &= 
\frac{\int_{\mathbb{R}^{M}}  [\vct{c}^\tra\hat{\vct{\theta}}]  f_0(\vct{x}+ \mat{G}\vct{h} -  \mat{G}\hat{\vct{\theta}}) \ \tm{d}\hat{\vct{\theta}}}
{\int_{ \mathbb{R}^{M}}   f_0(\vct{x} + \mat{G}\vct{h} - \mat{G}\hat{\vct{\theta}}) \ \tm{d}\hat{\vct{\theta}}}\\
&= \frac{\int_{\mathbb{R}^{M}}  [\vct{c}^\tra\hat{\vct{\theta}}]  f(\vct{x}| \hat{\vct{\theta}} - \vct{h} ) \ \tm{d}\hat{\vct{\theta}}}
{\int_{ \mathbb{R}^{M}}   f(\vct{x} | \hat{\vct{\theta}}-\vct{h} ) \ \tm{d}\hat{\vct{\theta}}}\\
&= 
\frac{\int_{\mathbb{R}^{M}}  [\vct{c}^\tra\hat{\vct{\theta}}]  f(\vct{x}| \hat{\vct{\theta}}) \ \tm{d}\hat{\vct{\theta}}}
{\int_{ \mathbb{R}^{M}}   f(\vct{x} | \hat{\vct{\theta}}) \ \tm{d}\hat{\vct{\theta}}} + \vct{c}^\tra \vct{h}\\
&= g^*(\vct{x}) + \vct{c}^\tra\vct{h}
\end{align}

\item For any choice of prior distribution $p(\vct{\theta})$, any estimator $g(\vct{x})$ of $\vct{c}^T \vct{\theta}$ satisfies
\begin{align}
\mathcal{M}(g(\vct{x})) & \ge  \sup_{p(\vct{\theta})}  \mathcal{B}(g(\vct{x}),p(\vct{\theta})) \\
& \ge \sup_{p(\vct{\theta})}   \inf_{\tilde{g}(\vct{x})}\mathcal{B}(\tilde{g}(\vct{x}),p(\vct{\theta})) = \mathcal{B}_{0} \label{eq:minimaxbayeslb}
\end{align}
Further, it can be proved (by contradiction) that $\mathcal{M}(g(\vct{x})) = \mathcal{B}_{0}$ holds if and only if $g(\vct{x})$ is minimax. 
Now consider the estimator $g^*(\vct{x})$ of (\ref{eq:minimaxoptestimator}). From (\ref{eq:minimaxbayeslb}), we already have $\mathcal{M}(g^*(\vct{x})) \ge  \mathcal{B}_{0}$. We shall now show that $\mathcal{B}_{0} \ge \mathcal{M}(g^*(\vct{x})) $ also holds, hence proving that  $\mathcal{B}_{0} = \mathcal{M}(g^*(\vct{x})) $, and thus that $g^*(\vct{x})$ is minimax. 

Consider a sequence of prior distributions $p_i(\vct{\theta})$, each uniformly distributed over a support set $\mat{\Theta}_i$ for $i=1,2,\cdots$, where
\begin{align}
 \mat{\Theta}_i = \left\{ \vct{\theta} : (-i)\cdot \vct{1}_M \le \vct{\theta} \le i \cdot \vct{1}_M\right\}
\end{align}
Here the inequality $(-i)\cdot \vct{1}_M \le \vct{\theta} \le i \cdot \vct{1}_M$ implies that all the elements of the vector $\vct{\theta}$ lie in the interval $[-i,i]$.
Given a prior distribution $p_i(\vct{\theta})$, the estimator that minimizes  $\mathcal{B}(g(\vct{x}),p(\vct{\theta})) $ is the minimum mean square error (MMSE) estimator,
\begin{align}
g_i(\vct{x}) = 
\int_{\vct{\theta} \in  \mat{\Theta}_i}  [\vct{c}^\tra\vct{\theta} ] f_i( \vct{\theta} |\vct{x}) \ \tm{d}\vct{\theta}
\end{align}
where $ f_i( \vct{\theta} |\vct{x})$ represents the posterior pdf 
\begin{align}
f_i( \vct{\theta} |\vct{x}) &= \frac{f(\vct{x}| \vct{\theta}) p_i(\vct{\theta})}{\int_{\tilde{\vct{\theta}} \in \mat{\Theta}_i}   f(\vct{x}|\tilde{\vct{\theta}}) p_i(\tilde{\vct{\theta}})\ \tm{d}\tilde{\vct{\theta}}}
= \frac{f(\vct{x}| \vct{\theta}) }{\int_{\tilde{\vct{\theta}} \in \mat{\Theta}_i}   f(\vct{x}|\tilde{\vct{\theta}}) \ \tm{d}\tilde{\vct{\theta}}}
\end{align}
Hence we can write 
\begin{align}
\mathcal{B}_{0} &= \sup_{p(\vct{\theta})}   \inf_{\tilde{g}(\vct{x})}\mathcal{B}(\tilde{g}(\vct{x}),p(\vct{\theta})) \nn\\
&\ge \inf_{\tilde{g}} \mathcal{B}(\tilde{g}(\vct{x}),p_i(\vct{\theta})) = \mathcal{B}(g_i(\vct{x}),p_i(\vct{\theta})) 
\label{eq:bayesupperlower}
\end{align}
Further, since $f_0(\vct{x})$ has finite support, we have 
\begin{align}
\lim_{i\rightarrow \infty } g_i(\vct{x}) &= \lim_{i\rightarrow \infty } 
\frac{\int_{\vct{\theta} \in \mat{\Theta}_i}  [\vct{c}^\tra\vct{\theta} ] f(\vct{x}| \vct{\theta}) \ \tm{d}\vct{\theta}}
{\int_{\vct{\theta} \in \mat{\Theta}_i}   f(\vct{x}| \vct{\theta}) \ \tm{d}\vct{\theta}} \nn \\
&= \frac{\int_{\vct{\theta} \in \mat{\Theta}(\vct{x})}  [\vct{c}^\tra\vct{\theta} ] f_0(\vct{x} - \mat{G} \vct{\theta}) \ \tm{d}\vct{\theta}}
{\int_{\vct{\theta} \in \mat{\Theta}(\vct{x})}   f_0(\vct{x} - \mat{G} \vct{\theta}) \ \tm{d}\vct{\theta}} 
= g^*(\vct{x})\nn
\end{align}
where
\begin{align}
\mat{\Theta}(\vct{x}) = 
\left\{ 
 \vct{\theta} : 
(\vct{x} - \mat{G} \vct{\theta}) > 0 
\ \ \tm{and}\ \  
(\vct{x} - \mat{G} \vct{\theta}) < L\cdot \vct{1}_{N} 
\right\}
\end{align}
and hence 
\begin{align}
&\lim_{i\rightarrow \infty } \mathcal{B}(g_i(\vct{x}),p_i(\vct{\theta}))  \nn \\
&= \lim_{i \rightarrow \infty } \mathcal{B}(g^*(\vct{x}),p_i(\vct{\theta})) \\
&= \lim_{i \rightarrow \infty } \mathcal{M}(g^*(\vct{x})) \quad \tm{(Since $g^*(\vct{x})$ is shift invariant)}\\
&= \mathcal{M}(g^*(\vct{x})) \label{eq:maplimitminimax}
\end{align}
From (\ref{eq:bayesupperlower}) and (\ref{eq:maplimitminimax}), we obtain $\mathcal{B}_{0} \ge \mathcal{M}(g^*(\vct{x}))$, hence completing the proof.
\item Since $g^*(\vct{x})$ is shift invariant, from Lemma \ref{lemma:shiftinv_lemma} we have 
$\mathcal{R}(g(\vct{x}),\vct{\theta}) = \mathcal{M}(g(\vct{x}))$.
Further, since all shift invariant estimators have constant conditional MSE, and $g^*(\vct{x})$ minimizes $\mathcal{M}(g(\vct{x}))$, it also minimizes $\mathcal{R}(g(\vct{x}),\vct{\theta})$ for every value of $\vct{\theta}$.
\item We shall prove that $g^*(\vct{x})$ is unbiased by contradiction. Assume $g^*(\vct{x})$ is biased. 
Since $g^*(\vct{x})$ is shift invariant, its bias should be constant with respect to $\vct{\theta}$. Let 
\begin{align}
\beta = \Ebr{ \left[ g^*(\vct{x}) - \vct{c}^\tra \vct{\theta } \right] \ |\ \vct{\theta}} = \Ebr{g^*(\vct{x})\ |\ \vct{\theta}=_M}
\end{align}
denote this constant bias. Now consider a new estimator of $\vct{c}^\tra \vct{\theta}$, given as
\begin{align}
\hat{g}(\vct{x}) = g^*(\vct{x}) - \beta 
\end{align}
It is easy to show that $\hat{g}(\vct{x})$ is also shift invariant. 
Further,
\begin{align}
\mathcal{M}\left(\hat{g}(\vct{x})\right) &=
\Ebr{\left[\hat{g}(\vct{x}) - \vct{c}^T\vct{\theta}\right]^2 \ |\ \vct{\theta}}\\
&= \Ebr{\left[g^*(\vct{x}) - \vct{c}^T\vct{\theta}\right]^2 \ |\ \vct{\theta}} \\
&\ \  - 2\beta  \Ebr{\left[g^*(\vct{x}) - \vct{c}^T\vct{\theta}\right] \ |\ \vct{\theta}} + \beta^2\\
&= \mathcal{M}\left(g^*(\vct{x})\right)  - \beta^2  \ \  < \ \  \mathcal{M}\left(g^*(\vct{x})\right)  
\end{align}
However, this is impossible since $g^*(\vct{x})$ has already been shown to minimize the maximum MSE. Thus, the assumption that $g^*(\vct{x})$ is biased is incorrect. 
\end{enumerate}
\end{proof}
$ $ \\
\begin{proof}[Proof of Lemma \ref{lemma:linearsummation}]
Using theorem \ref{theorem:MinimaxEstimator}, we obtain 
\begin{align}
g_i^*(\vct{x}) = 
\frac{\int_{\mathbb{R}^{M}}  \hat{\theta}_i f(\vct{x}| \hat{\vct{\theta}}) \ \tm{d}\hat{\vct{\theta}}}
{\int_{ \mathbb{R}^{M}}   f(\vct{x} | \hat{\vct{\theta}}) \ \tm{d}\hat{\vct{\theta}}}
\end{align}
and
\begin{align}
g^*(\vct{x}) &= 
\frac{\int_{\mathbb{R}^{M}}  [\vct{c}^\tra\hat{\vct{\theta}}]  f(\vct{x}| \hat{\vct{\theta}}) \ \tm{d}\hat{\vct{\theta}}}
{\int_{ \mathbb{R}^{M}}   f(\vct{x} | \hat{\vct{\theta}}) \ \tm{d}\hat{\vct{\theta}}} \\
&=\frac{\sum_{i=1}^{M} c_i \int_{\mathbb{R}^{M}}  \theta_i  f(\vct{x}| \hat{\vct{\theta}}) \ \tm{d}\hat{\vct{\theta}}}
{\int_{ \mathbb{R}^{M}}   f(\vct{x} | \hat{\vct{\theta}}) \ \tm{d}\hat{\vct{\theta}}} = \sum_{i=1}^{M} c_i g_i^*(\vct{x})
\end{align}
hence proving the theorem.
\end{proof}

\begin{proof}[Proof of Theorem \ref{theorem:maxmselinearity}]
Consider the problem of estimating $\vct{c}^\tra \vct{\theta}$ given all the observations $\vct{x}_1, \cdots, \vct{x}_K$. It is easy to show that this problem is a vector location parameter problem as defined in Section \ref{sec:minimaxest}. Hence, a minimax optimum estimator for this problem can be obtained via Theorem \ref{theorem:MinimaxEstimator}.
Denote this minimax estimator as $h^*(\vct{x}_1, \cdots, \vct{x}_K)$. 
Further, note that $h^*(\vct{x}_1, \cdots, \vct{x}_K)$ and $g^*(\vct{x})$ are both estimators of $\vct{c}^\tra \vct{\theta}$, but $h^*(\vct{x}_1, \cdots, \vct{x}_K)$ has more information available to it, since $\sum_{k=1}^{K}\vct{x}_k = \vct{x}$.
Hence, we must have
\begin{align}
\mathcal{M}(g^*(\vct{x})) \ge \mathcal{M}(h^*(\vct{x}_1, \cdots, \vct{x}_K)) \label{eq:linsumpart0}
\end{align}
Further, using Lemma \ref{lemma:linearsummation}, it can be shown that
\begin{align}
h^*(\vct{x}_1, \cdots, \vct{x}_K) = \sum_{k=1}^{K} h_k^*(\vct{x}_k)
\end{align}
Since minimax estimators are shift invariant, we can write
\begin{align}
&\mathcal{M}(h^*(\vct{x}_1, \cdots, \vct{x}_K)) \\
&= \sup_{\vct{\theta} \in \mathbb{R}^N} \Ebr{\left[ h^*(\vct{x}_1, \cdots, \vct{x}_K) - \vct{c}^T \vct{\theta}\right]^2 \Big|\ \vct{\theta}}\\
&= \Ebr{\left[ h^*(\vct{x}_1, \cdots, \vct{x}_K) \right]^2 \Big|\ \vct{\theta} = \vct{0}_M} 
\end{align}
This can be further simplified as
\begin{align}
& \mathcal{M}(h^*(\vct{x}_1, \cdots, \vct{x}_K))\\
 &= \Ebr{\left[\sum_{k=1}^{K} h_k^*(\vct{x}_k) \right]^2 \Big|\ \vct{\theta}_1 = \vct{0}_M, \cdots, \vct{\theta}_K = \vct{0}_M} \\
&= \sum_{k=1}^{K}\Ebr{ [h_k^*(\vct{x}_k)]^2\Big|\ \vct{\theta}_k = \vct{0}_M}  \nn 
 + \sum_{k_1=1}^{K} \sum_{
\begin{subarray}
k k_2=1\\
k_2 \ne k_1
\end{subarray}}^{K} \Bigg[ \\
& \quad \Ebr{ h_{k_1}^*(\vct{x}_{k_1}) h_{k_2}^*(\vct{x}_{k_2})\ \Big|\ \vct{\theta}_{k_1} = \vct{0}_M , \vct{\theta}_{k_2} = \vct{0}_M} \Bigg] 
\end{align}
We note that $h_1^*(\vct{x}_1), \cdots, h_K^*(\vct{x}_K)$ are all mutually independent conditioned on the unknown parameters, due to our initial assumption that $\vct{x}_1, \cdots, \vct{x}_K$ are mutually independent as per (\ref{eq:indepdenceassumption}). Hence, we obtain
\begin{align}
& \mathcal{M}(h^*(\vct{x}_1, \cdots, \vct{x}_K))\\
&= \sum_{k=1}^{K} \Ebr{[h_k^*(\vct{x}_k)]^2\Big|\ \vct{\theta}_{k} = \vct{0}_M}  \nn \\
&\ \ + \sum_{k_1=1}^{K} \sum_{
\begin{subarray}
k k_2=1\\
k_2 \ne k_1
\end{subarray}}^{K} \Ebr{ h_{k_1}^*(\vct{x}_{k_1})\Big|\ \vct{\theta}_{k_1} = \vct{0}_M} \nn \\
&\qquad\qquad \qquad \qquad \cdot  \Ebr{ h_{k_2}^*(\vct{x}_{k_2})\Big|\ \vct{\theta}_{k_2} = \vct{0}_M}  \label{eq:linsumpart1}
\end{align}
Since $h_{k}^*(\vct{x}_{k})$ is a minimax estimator, it is unbiased
and shift invariant according to Theorem \ref{theorem:MinimaxEstimator}, and hence 
\begin{align}
& \Ebr{ h_{k}^*(\vct{x}_{k})\Big|\ \vct{\theta} = \vct{0}_M} = 0 \ , \label{eq:linsumpart2}\\
&\Ebr{ [h_{k_1}^*(\vct{x}_{k_1})]^2 \Big|\ \vct{\theta}_{k_1} = \vct{0}_M} =  \mathcal{M}(h_k^*(\vct{x}_k)) \label{eq:linsumpart3}
\end{align}
From (\ref{eq:linsumpart1}), (\ref{eq:linsumpart2}) and (\ref{eq:linsumpart3}) we obtain 
\begin{align}
\mathcal{M}(h^*(\vct{x}_1, \cdots, \vct{x}_K)) = \sum_{k=1}^{K}  \mathcal{M}(h_k^*(\vct{x}_k)) \label{eq:linsumpart4}
\end{align}
Finally, from (\ref{eq:linsumpart0}) and (\ref{eq:linsumpart4}), we obtain 
\begin{align}
\mathcal{M}(g^*(\vct{x})) \ge \sum_{k=1}^{K}  \mathcal{M}(h_k^*(\vct{x}_k))
\end{align}
hence concluding the proof.
\end{proof}

\begin{proof}[Proof of Corollary \ref{theorem:msevsswitches}]
We shall prove this corollary by applying Theorem \ref{theorem:maxmselinearity} to POE under the S-model.
To this end, assume $N=KL$, where $K$ and $L$ are both integers.
For the $N$-node network, assuming $P$ pairs of timestamp differences are collected per the S-model, the observation vector can be written similar to  (\ref{eq:standard_obsmodel_vct}), as
\begin{align}
\vct{y} = d\vct{1}_{2P} + \delta \vct{e} + \vct{w} 
\end{align}
where $d$ and $\delta$ represent the unknown fixed delay and phase offset, while $\vct{w}$ represents the $2P\times 1$ vector of queuing delays. 

Now suppose that the cascade of $N=KL$ nodes is split into $K$ smaller cascades, each consisting of $L$ nodes.
Each cascade of $L$ nodes is placed between a new master-slave pair, resulting in $K$ new networks (see example in Fig. \ref{fig:MinimaxMSE_networks}).
\begin{figure}
\centering
\subfloat[Original network] {\includegraphics[width=0.75\columnwidth]{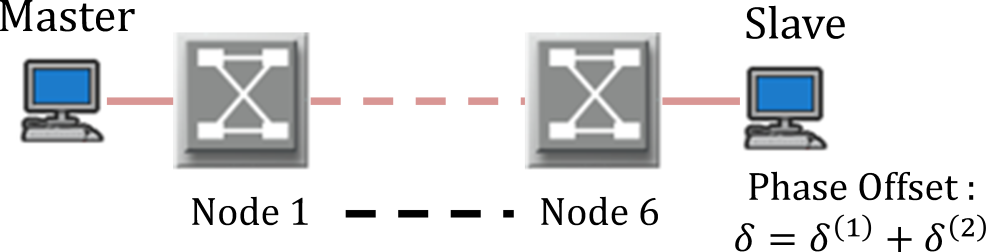}\label{fig:MinimaxMSE_6node_network}}\\
\subfloat[Networks obtained after splitting] {\includegraphics[width=0.75\columnwidth]{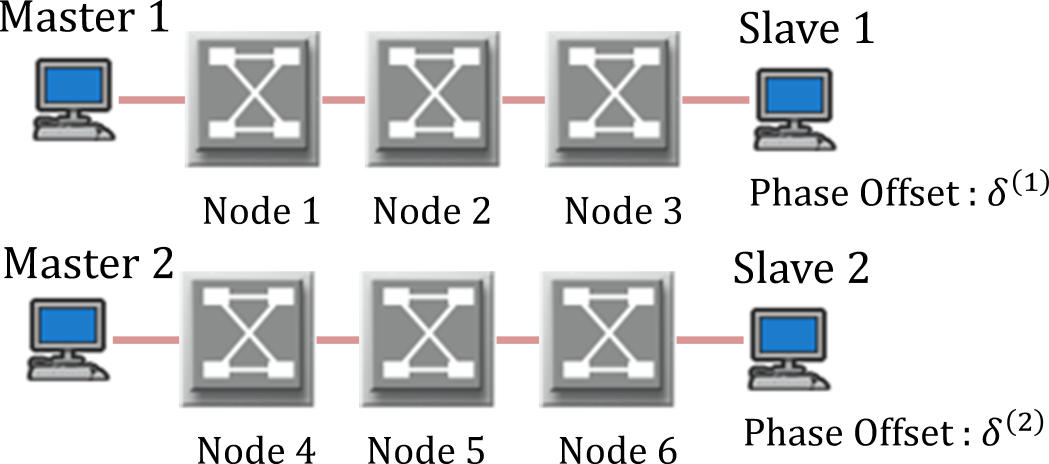}\label{fig:MinimaxMSE_2x3node_network}}
\caption{Example of a network containing $N=6$  intermediate nodes, that has been split into $K=2$ networks, each containing $L=3$ intermediate nodes.}\label{fig:MinimaxMSE_networks}
\end{figure}
Let the phase offset of the slave  in the $k^\tm{th}$ network be $\delta^{(k)}$, and let the fixed delay in the $k^\tm{th}$ network be $d^{(k)}$.  
Assume that the phase offsets and fixed delays satisfy the relation 
\begin{align}
\sum_{k=1}^{K}\delta^{(k)} = \delta , \ \ \sum_{k=1}^{K} d^{(k)} = d  \label{eq:dkdeltaksum}
\end{align}
Assuming that $P$ observation pairs are collected per the S-model, the observation vector for each $L$-node network can be written, similar to (\ref{eq:standard_obsmodel_vct}), as 
\begin{align}
\vct{y}^{(k)} = d^{(k)} \vct{1}_{2P} + &\delta^{(k)} \vct{e} + \vct{w}^{(k)} 
\end{align}
for $k=1,\cdots,K$. Here $\vct{w}^{(k)}$ represents the $2P\times 1$ vector of queuing delays in the $k^\tm{th}$ network. 
Since the single-node queuing delays across all nodes are identically distributed, the minimax MSE associated with estimating $\delta^{(k)}$ from $\vct{y}^{(k)}$ will be identical, and equal  $\minimaxmse(L)$ in all the $L$-node networks. 
Note that due to the shift invariance of the minimax estimator and the result
in Lemma \ref{lemma:shiftinv_lemma}, the minimax MSE will remain unchanged regardless
of the assumption in (\ref{eq:dkdeltaksum}), since the minimax MSE does not depend on the value of $\delta^{(k)}$ or $d^{(k)}$.
In order to apply Theorem \ref{theorem:maxmselinearity}, we note that the queuing delay vector under the $KL$ node network can be written as sum of the queuing delay vectors under each $L$-node network, i.e. $\vct{w} = \sum_{k=1}^{K}\vct{w}^{(k)}$. 
Further, due to the assumption that the single-node queuing delays are mutually independent, we have
\begin{align}
& f(\vct{y}^{(k_1)},\vct{y}^{(k_2)}|\delta^{(k_1)}, d^{(k_1)}, \delta^{(k_2)}, d^{(k_2)} ) \nn \\ &=f(\vct{y}^{(k_1)}|\delta^{(k_1)}, d^{(k_1)}) f(\vct{y}^{(k_2)}|\delta^{(k_2)}, d^{(k_2)}) 
 .
\end{align}
Due to the assumption in (\ref{eq:dkdeltaksum}), we also have 
\begin{align}
\vct{y} &= d\vct{1}_{2P} + \delta \vct{e} + \vct{w} \\
&=  \left[\sum_{k=1}^{K} d^{(k)}\right] \vct{1}_{2P} + \left[\sum_{k=1}^{K}\delta^{(k)}\right] \vct{e} +  \left[\sum_{k=1}^{K}\vct{w}^{(k)}\right]\\
&=  \sum_{k=1}^{K} \left[ d^{(k)} \vct{1}_{2P} + \delta^{(k)} \vct{e} +  \vct{w}^{(k)} \right]\ \ 
=\ \sum_{k=1}^{K}\vct{y}^{(k)}
\end{align}
Noting the similarity in the relationships between $\vct{y}$, $\vct{y}^{(k)}$ and the vectors $\vct{x}$, $\vct{x}_k$ in Theorem \ref{theorem:maxmselinearity}, we can apply Theorem \ref{theorem:maxmselinearity} to obtain the relation
\begin{align}
\minimaxmse(KL) \ge K \minimaxmse(L) 
\end{align}
which concludes the proof.	
\end{proof}

\begin{figure}[h]
\centering
\subfloat[TM1, $20-40\%$ Load] {\label{fig:pdf_TM1_20_40} \includegraphics[width=\perfcurveswidth]{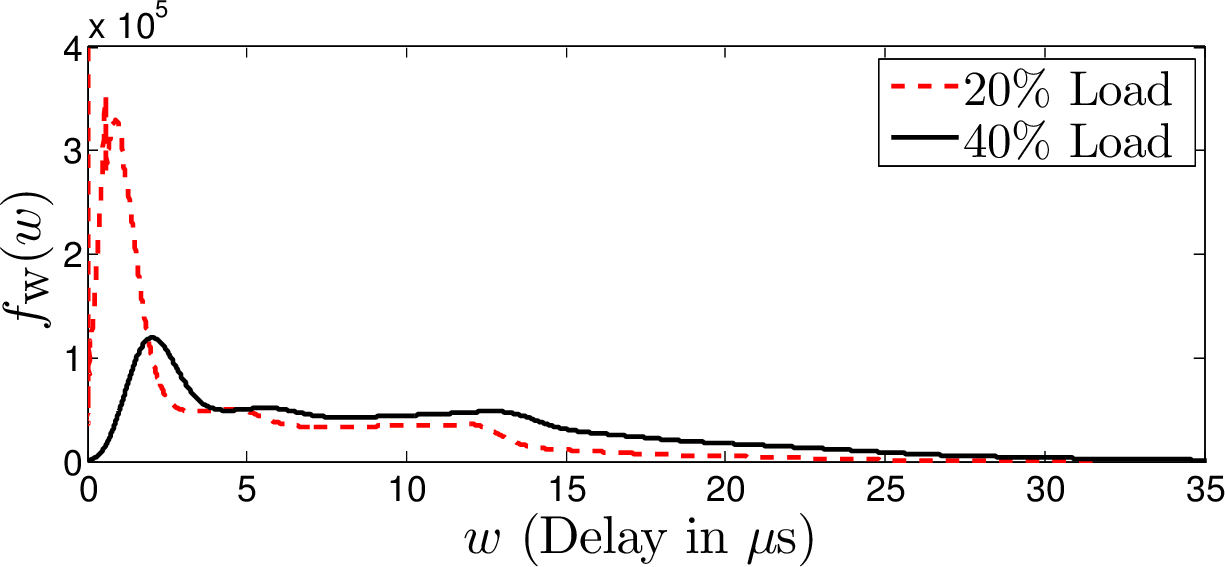} } \\ 
\subfloat[TM1, $60-80\%$ Load] {\label{fig:pdf_TM1_60_80} \includegraphics[width=\perfcurveswidth]{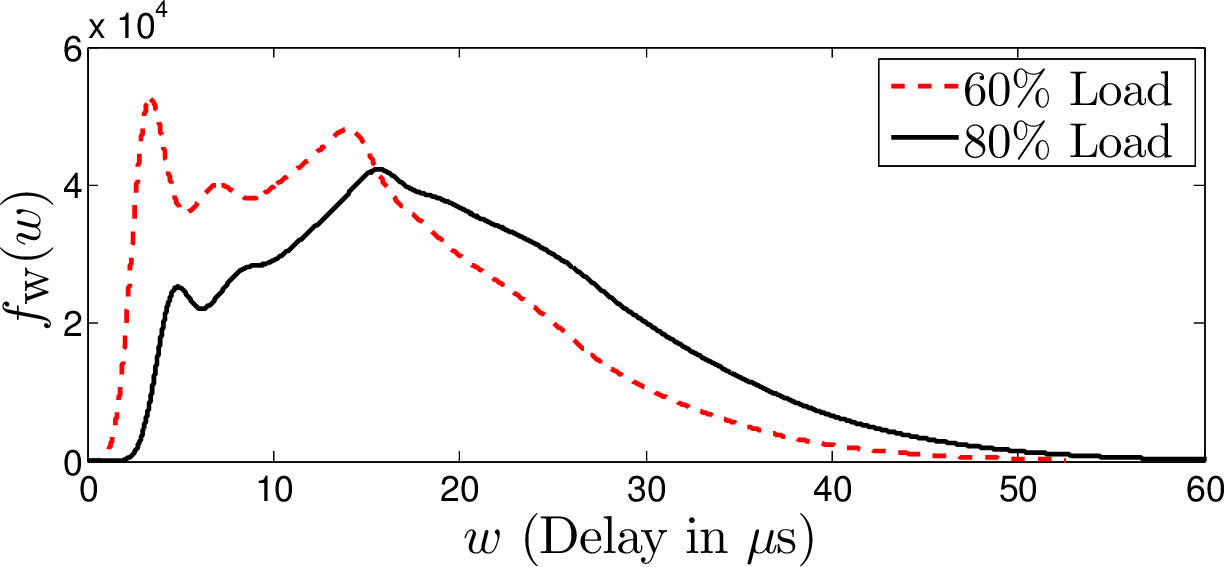} } \\ 
\subfloat[Mixed Traffic, with TM2 for cross traffic and uniform packet size distribution for inline traffic.] {\label{fig:pdf_Mixed}\includegraphics[width=\perfcurveswidth]{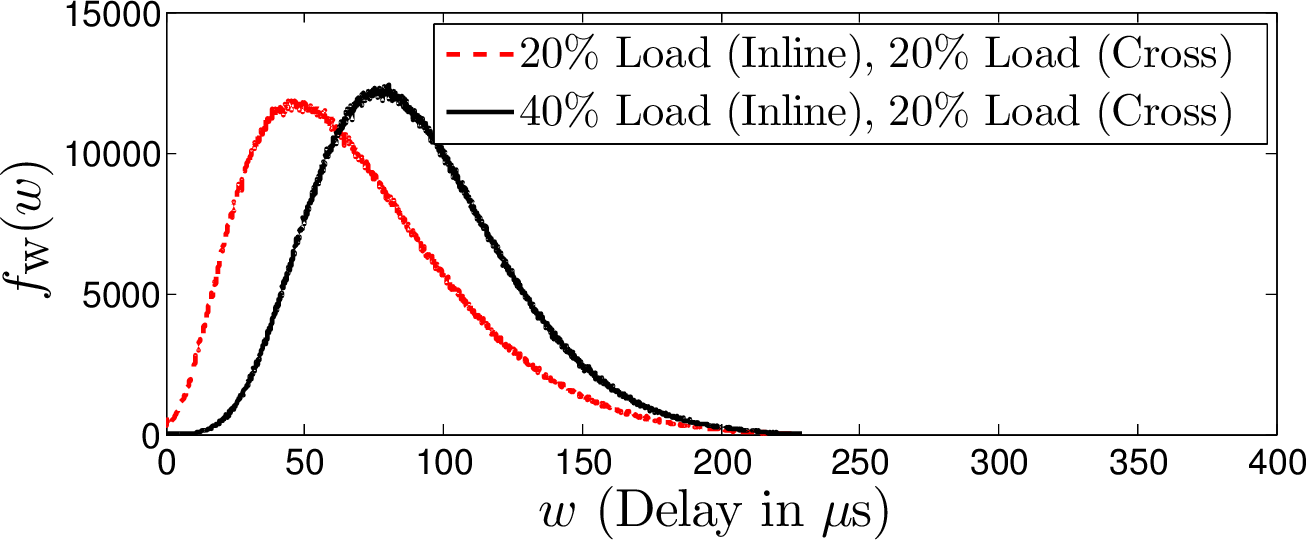} }
\caption{Plots of queuing delay distributions under different network conditions}
\label{fig:pdfs}
\end{figure}
\begin{figure}
\centering
\subfloat[TM1, 20\% Load] {\includegraphics[width=\perfcurveswidth]{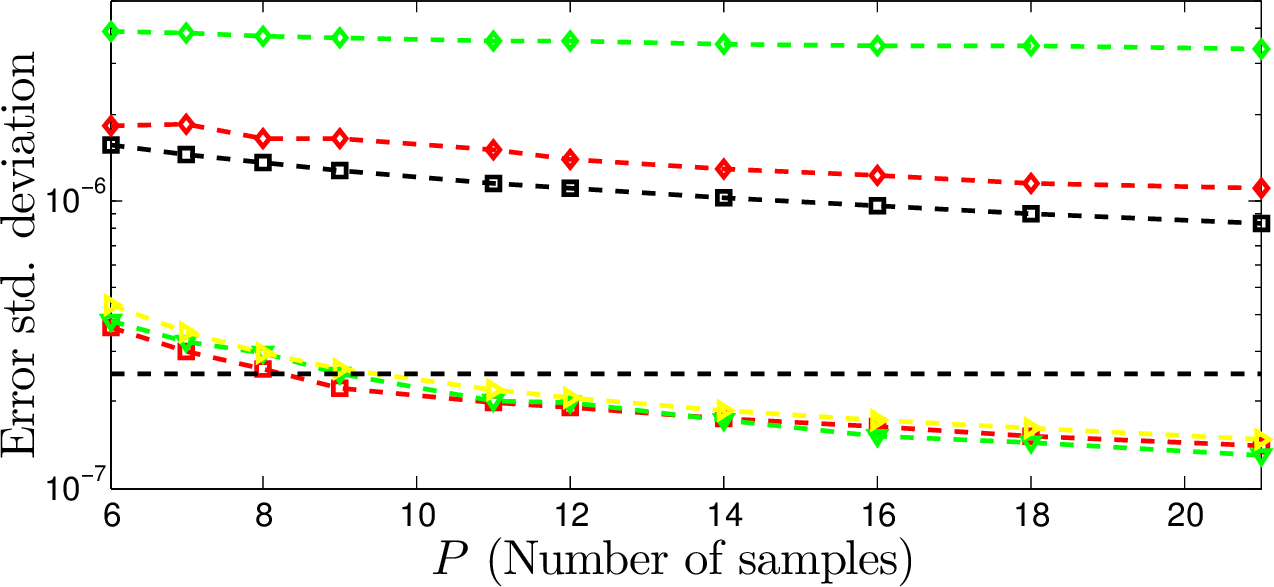} } \\ 
\subfloat[TM1, 40\% Load] {\includegraphics[width=\perfcurveswidth]{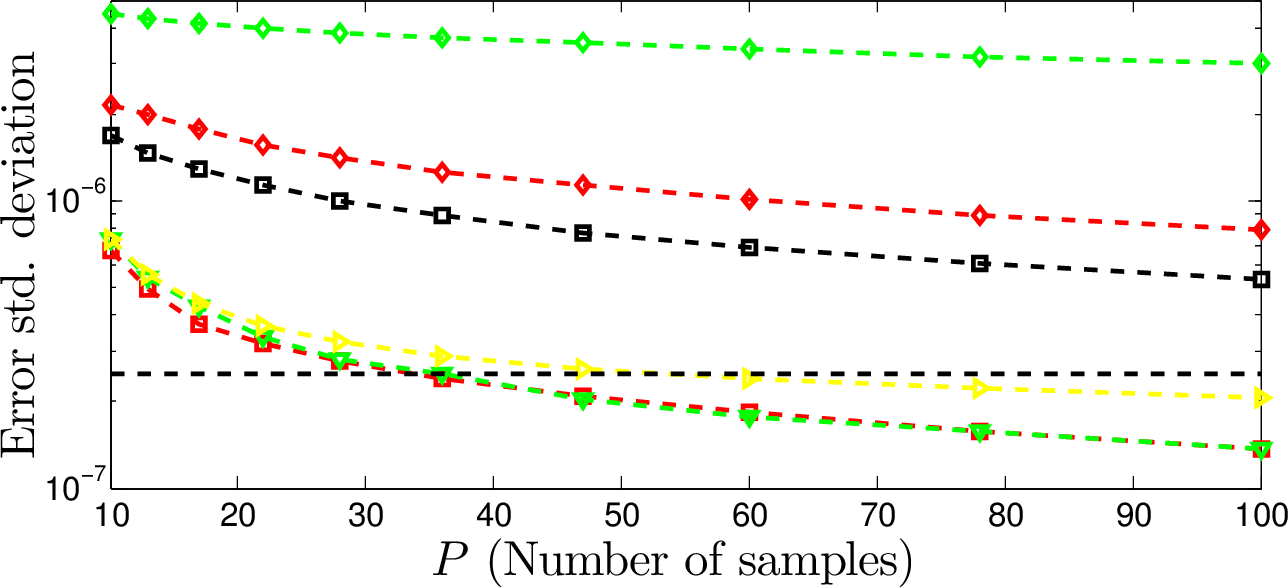}}\\
\subfloat[TM1, 60\% Load] {\includegraphics[width=\perfcurveswidth]{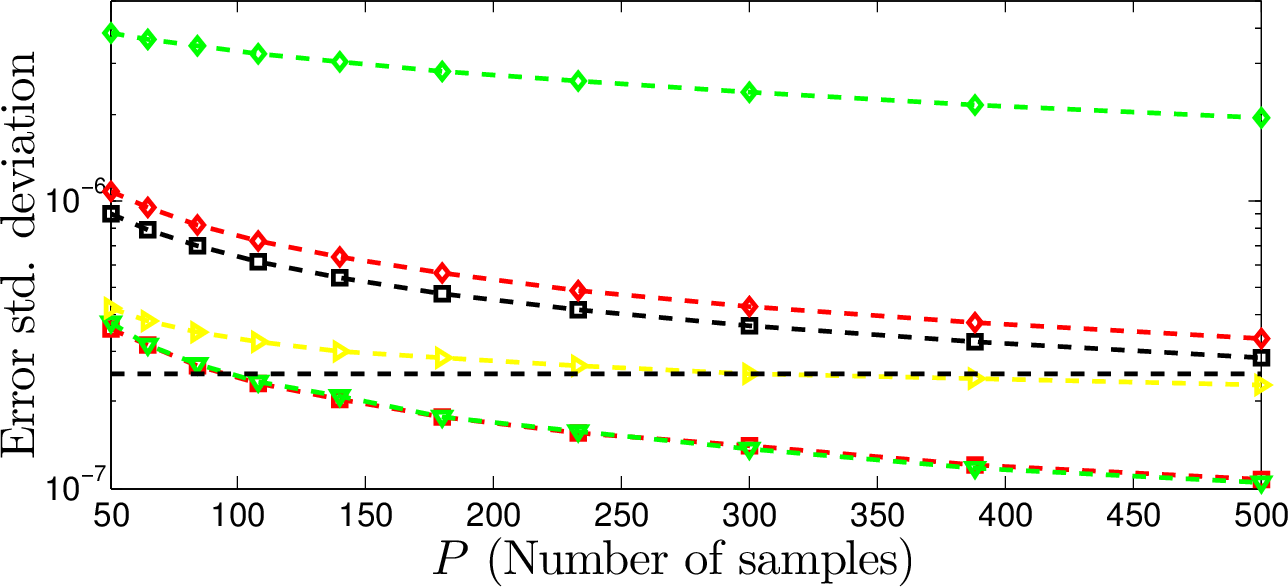} } \\ 
\subfloat[TM1, 80\% Load] {\includegraphics[width=\perfcurveswidth]{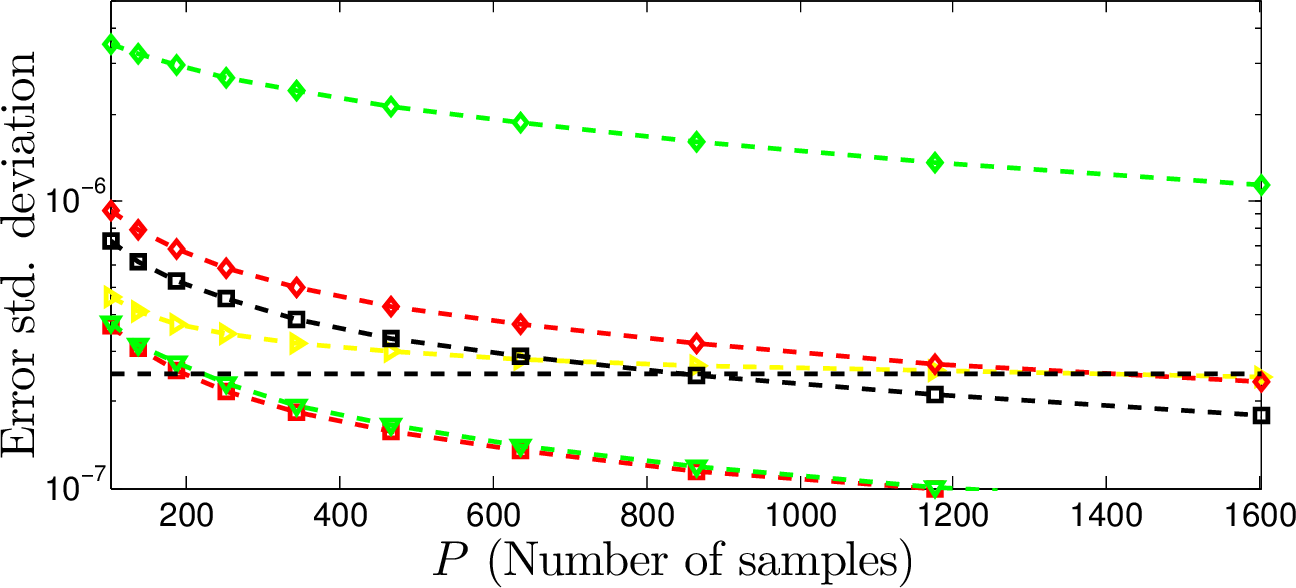}}\\
\subfloat{\includegraphics[width=\legendwidth]{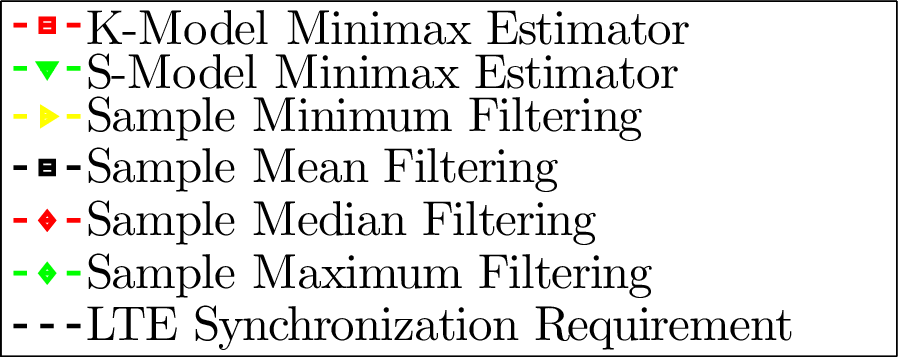}}
\caption{Performance comparison of different estimators under symmetric cross traffic.}
\label{fig:mse_symm_cross}
\end{figure}
\begin{figure}
\centering
\subfloat[20\% Load (Inline), 20\% Load (Cross, TM2)] {\includegraphics[width=\perfcurveswidth]{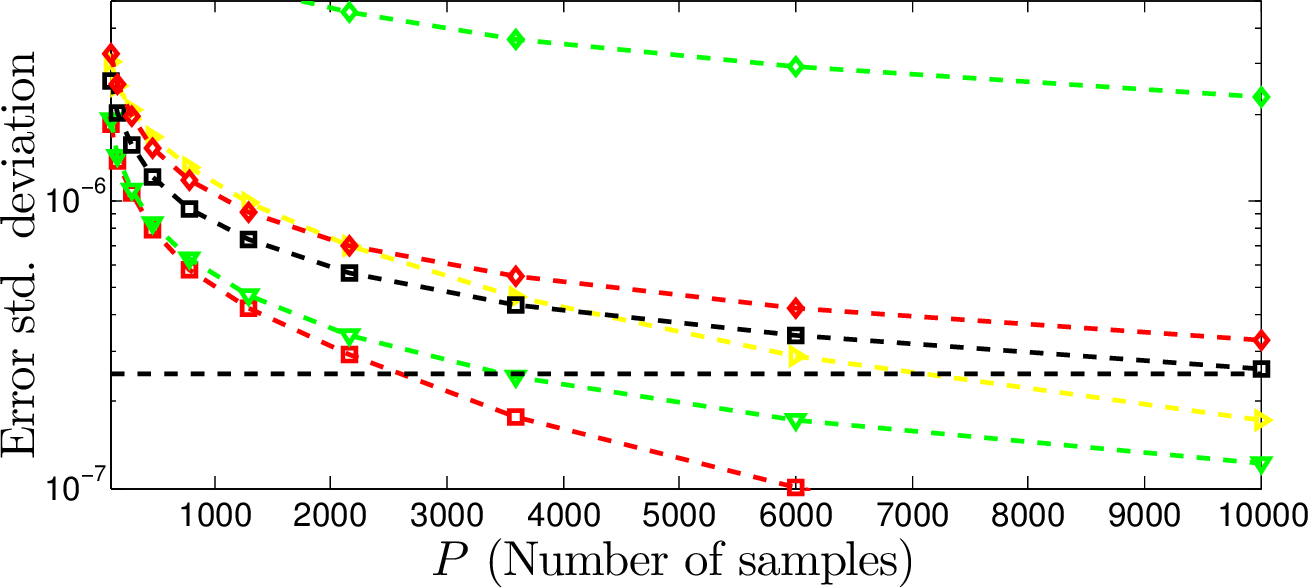}\label{fig:mse_symm_mixed_low}}\\
\subfloat[40\% Load (Inline), 20\% Load (Cross, TM2)] {\includegraphics[width=\perfcurveswidth]{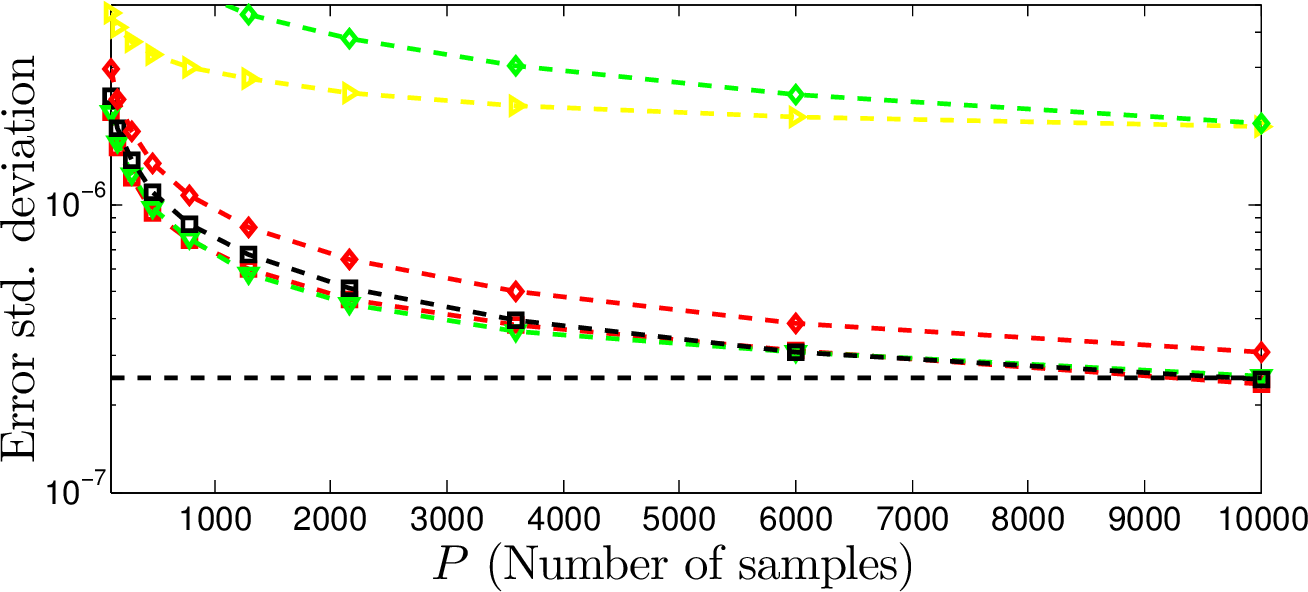}\label{fig:mse_symm_mixed_high}}\\
\subfloat{\includegraphics[width=\legendwidth]{Legend1.png}}
\caption{Performance comparison of different estimators under symmetric mixed traffic.}
\label{fig:mse_symm_mixed}
\end{figure}
\begin{figure}
\centering
\subfloat {\includegraphics[width=\perfcurveswidth]{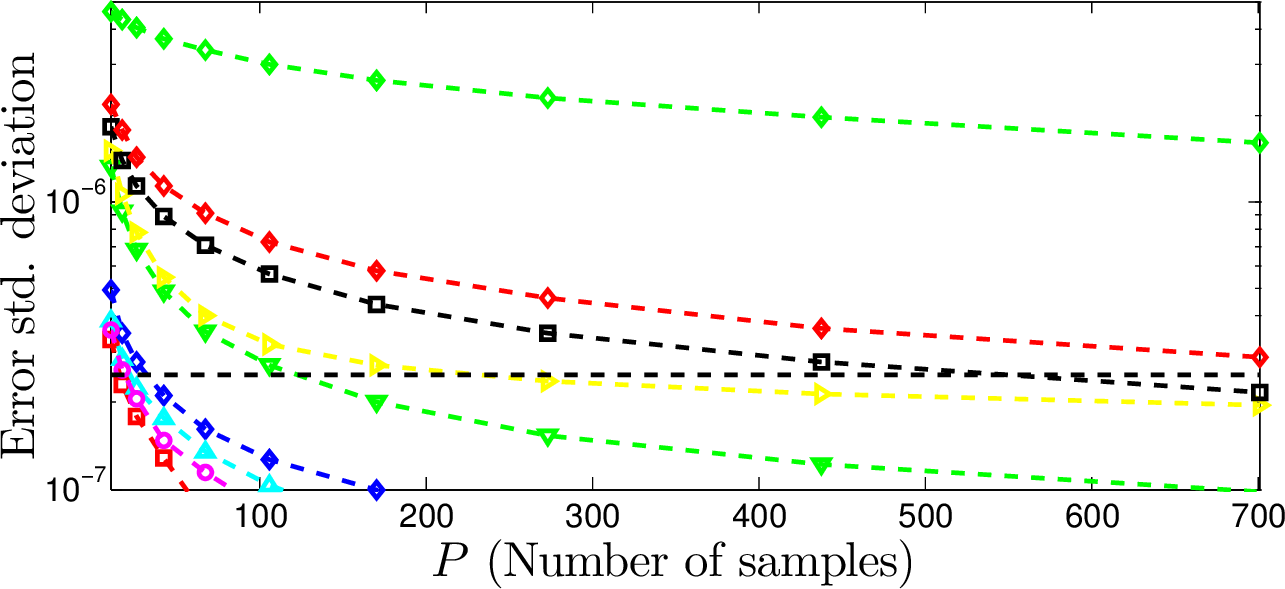}}\\
\subfloat{\includegraphics[width=\legendwidth]{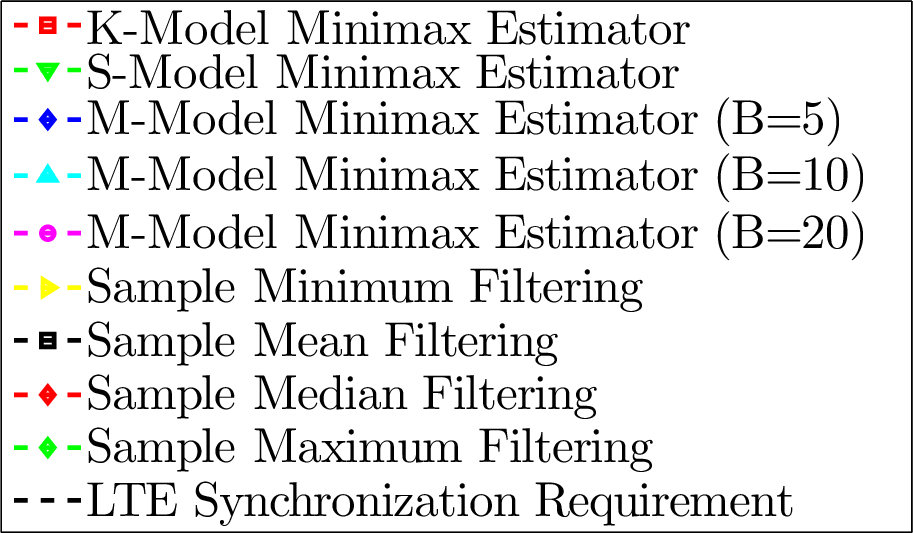}}
\caption{Performance comparison of different estimators under asymmetric cross traffic. Forward path: 80\% Load (TM1), Reverse path: 20\% Load (TM1).}
\label{fig:mse_asymm_cross}
\end{figure}
\begin{figure}
\centering
\subfloat{\includegraphics[width=\perfcurveswidth]{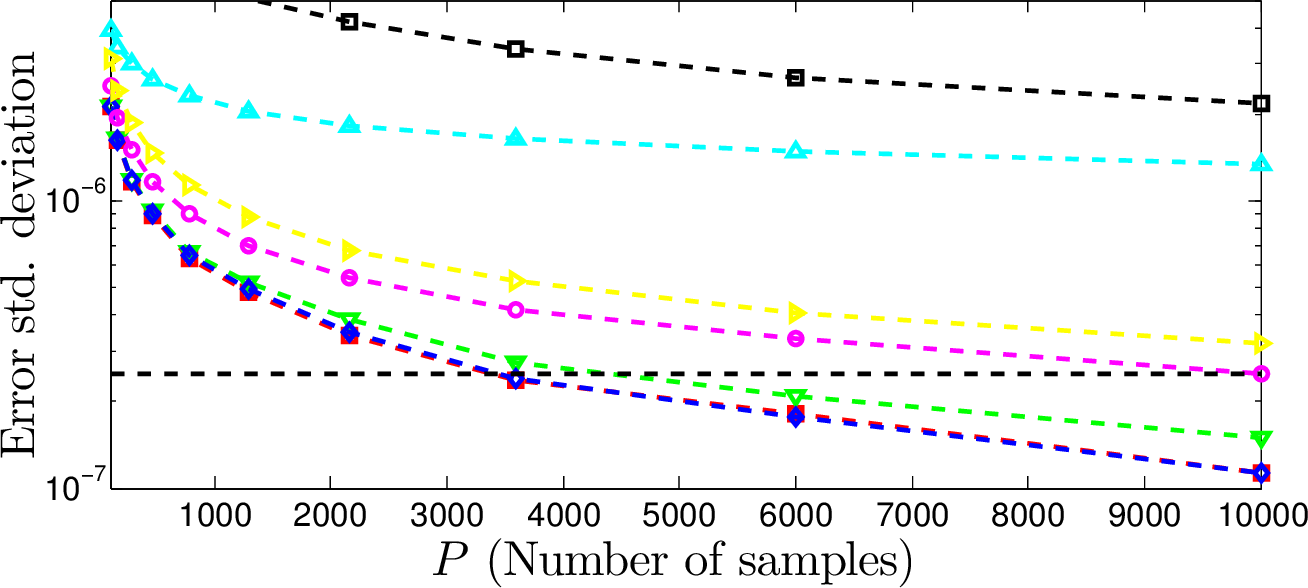}}\\
\subfloat{\includegraphics[width=\legendwidth]{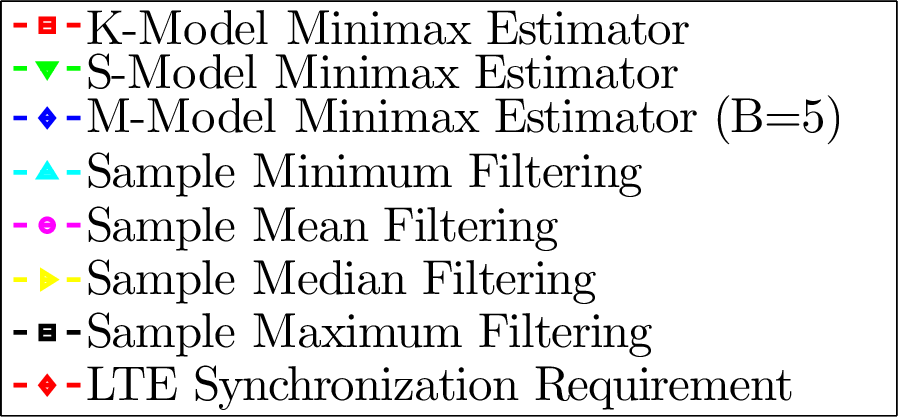}}
\caption{Performance comparison of different estimators under asymmetric mixed traffic. Traffic models used are TM2 for cross traffic and uniform packet size distribution for inline traffic. 
The forward path has $40\%$ inline load and $20\%$ cross load, while the reverse path has $20\%$ inline  load and $20\%$ cross load.}
\label{fig:mse_asymm_mixed}
\end{figure}


\end{document}